\documentclass[10pt,a4paper]{scrartcl}
\usepackage[margin=3.5cm]{geometry}
\usepackage[utf8]{inputenc}
\usepackage[mathscr]{euscript}

\usepackage[utf8]{inputenc}
\usepackage{amsbsy} 
\usepackage{amsmath}
\usepackage{amssymb}
\usepackage{amsthm}
\usepackage{enumerate}
\usepackage{caption}
\usepackage{subcaption}
\usepackage{dsfont}
\usepackage{import}
\usepackage{mathtools}
\usepackage{microtype}
\usepackage{xcolor}

\usepackage{xr}
\usepackage[unicode,pdfborder={0 0 0}]{hyperref}
\usepackage[capitalize,noabbrev]{cleveref}

\makeatletter
\def\namedlabel#1#2{\begingroup
   \def\@currentlabel{#2}%
   \label{#1}\endgroup
}
\makeatother

\mathtoolsset{centercolon}

\crefname{equation}{equation}{equations}%

\newcommand{\diff}{\,\mathrm d}
\newcommand{\indicator}{\mathds{1}}
\newcommand{\EE}{\mathds{E}}
\newcommand{\NN}{\mathds{N}}
\newcommand{\PP}{\mathds{P}}
\newcommand{\QQ}{\mathds{Q}}
\newcommand{\RR}{\mathds{R}}
\newcommand{\FF}{\mathds{F}}

\newcommand{\calS}{\mathscr{S}}
\newcommand{\calY}{\mathscr{Y}}

\def\P{{\mathbb P}}

\def\R{{\mathbb R}}

\def\F{{\mathcal F}}

\def\S{{\mathcal S}}

\def\U{{\mathcal U}}
\def\A{{\mathcal A}}
\def\K{{\mathcal K}}
\def\H{{\mathcal H}}
\def\O{{\mathcal O}}

\def\1{{\mathds 1}}

\DeclareOldFontCommand{\bf}{\normalfont\bfseries}{\mathbf}

\DeclarePairedDelimiterX{\closedStochasticInterval}[1]{[}{]}{\!\delimsize[#1\delimsize]\!}
\DeclarePairedDelimiterX{\leftOpenStochasticInterval}[1]{]}{]}{\!\delimsize]#1\delimsize]\!}
\DeclarePairedDelimiterX{\rightOpenStochasticInterval}[1]{[}{[}{\!\delimsize[#1\delimsize[\!}

\newcommand{\sgn}{\operatorname{sgn}}

\renewcommand{\rho}{(\text{\textcolor{red}{USE \texttt{\textbackslash{}correlationBW} OR \texttt{\textbackslash{}transformThetaToR} INSTEAD OF \texttt{\textbackslash{}rho}!}})}

\newcommand{\assetsProcess}{\Theta}
\newcommand{\baseS}{\bar S}
\newcommand{\assetbeta}{B}  

\newtheoremstyle{boldremark}
	{\topsep}   
	{\topsep}   
	{}          
	{}          
	{\bfseries} 
	{.}         
	{.5em}      
	{}          

\newtheorem{theorem}{Theorem}[section]
\newtheorem{proposition}[theorem]{Proposition}
\newtheorem{lemma}[theorem]{Lemma}
\newtheorem{corollary}[theorem]{Corollary}

\newtheorem{definition}[theorem]{Definition} 
\newtheorem{assumption}[theorem]{Assumption}

\theoremstyle{boldremark}

\newtheorem{remark}[theorem]{Remark}
\newtheorem{example}[theorem]{Example}

\numberwithin{equation}{section}

\newcommand{\ncomment}[1]{}

\usepackage[T1]{fontenc}
\usepackage{lmodern}

\newcommand{\revremove}[1]{{\color{red}}} 

\usepackage[draft]{changes}  

\newcommand{\ignore}[1]{}

\title{
Hedging with physical or cash settlement 
	under transient multiplicative price impact
}

\author{%
	Dirk Becherer, 
	Todor Bilarev\footnote{Support by German Science foundation DFG 
is gratefully acknowledged.} 
	\footnote{Email addresses: becherer(at)math.hu-berlin.de, bilarev.todor(at)gmail.com.
Humboldt-Universität zu Berlin, and FactSet Research Systems Inc.. Disclaimer: The opinions expressed in this publication are those of the authors. They do not purport to reflect views of their institutions. 
	}
}

\begin{document}
\maketitle

\begin{abstract}
		We solve the superhedging problem for European options in an illiquid extension of the Black-Scholes model,
		in which transactions have transient 
	price impact and the costs and the strategies for hedging are affected by physical or cash settlement requirements at maturity. Our analysis is based on a convenient choice of reduced effective coordinates of magnitudes at liquidation for geometric dynamic programming. 
		The price impact is transient over time and 
		 multiplicative, ensuring non-negativity of underlying asset prices while maintaining an arbitrage-free model. The basic (log-)linear example is a Black-Scholes model with relative price impact being proportional to the volume of shares traded, where the transience for impact on log-prices is being modelled like in Obizhaeva-Wang \cite{ObizhaevaWang13} for nominal prices. More generally, we allow for non-linear price impact and resilience functions. 
		 The viscosity solutions describing the minimal superhedging price are governed by the transient character of the price impact and by the physical or cash settlement specifications. 
		 Pricing equations under illiquidity extend    no-arbitrage pricing a la Black-Scholes for complete markets in a non-paradoxical way (cf.\ {\c{C}}etin, Soner and Touzi \cite{CetinSonerTouzi10})  even without additional frictions, and can recover it in base cases.
\end{abstract}

{\textbf{Keywords}: Transient price impact, multiplicative impact, hedging, option settlement,
	resilience, 
	viscosity solution, 
	geometric dynamical programming, effective coordinates.
}

{\textbf{MSC2010 subject classifications}: 
49L20, 49L25, 60H30, 91G20, 93E20 
}

{\textbf{JEL subject classifications}: C61, G12, G13}
 
\section{Introduction}\label{sec:Intro}

By using methods of stochastic target problems \cite{SonerTouzi02} and  geometric dynamic programming in suitably choosen reduced effective coordinates of magnitudes at liquidation,  we solve the superhedging problem for European derivatives in a market model with multiplicative transient price impact. 
If the market for the underlying is illiquid or because large volumes are to be traded, there is price impact and feedback effects from hedging 
can affect the minimal superhedging prices \cite{Frey1998,SchoenbucherWilmott2000,FreyPolte11,BankBaum04} and respective hedging strategies, which almost-surely super-replicate the option.
Since trades at maturity can alter the price of the underlying and thereby the derivative payout,  settlement specifications for the option in cash or physical units become relevant and this should show in pricing and hedging equations. 
As our results address hedging in terms of liquidation values, that means ``real'' instead of ``paper'' values \cite{Jarrow92}, we recover such effects, whereas \cite{Frey1998,
	FreyPolte11,BLZ16b} study hedging in terms of book (``paper'') values. The settlement constraints imposed for hedging (in \cref{sec:impact model}) in combination with stability by a suitably chosen notion of value, which depends continuously on trading strategies, moreover help to avoid some known paradoxical effects in price impact modelling (see \cite{BankBaum04,CetinSonerTouzi10}, \cite[Rmk.3.3]{BechererBilarevFrentrup-model-properties} and the comments after \eqref{eq:inst liq value}) and overly excessive opportunities of manipulating derivative payoffs (as in \cite[Sect.4.1]{SchoenbucherWilmott2000}).

The best known model for transient price impact is probably the one due to Obizhaeva and Wang \cite{ObizhaevaWang13}. 
It states that the dynamic holdings $\Theta$ of a large trader 
have additive linear impact (with parameter $\lambda>0$) on the prevailing prices $s$ of the underlying asset, via
\begin{align}\label{eqn:OWadditive}
	\begin{split}
		\textrm{(log-)price:} \qquad &d s_t =  d\bar{s}_t +\lambda dY_t\,, \quad \text{with} \\ 
		\textrm{impact level:} \qquad &dY_t = -\beta Y_t \, dt + d\Theta_t =: -h(Y_t)dt + d\Theta_t\,, 	
	\end{split} 
\end{align}
where $\bar s$ is a given unaffected (fundamental) price evolution for the underlying, while
$Y$ is a market impact level process driven by $\Theta$, whose
mean-reverting dynamics is 
linear in the 
asset holdings $\Theta$ of the large trader and transient over time, recovering at some resilience rate, given by the parameter $\beta>0 $ in the linear resilience function $h$.

Assuming price impact to be additive helps for mathematical tractability (i.p.\  if $\bar{s}$ is a martingale) and can simply serve to approximate multiplicative impact on short horizon. This is  common in the literature on optimal trade execution, as explained by Busseti and Lillo 
in \cite[Sect.6]{BussetiLillo12}, who further describe [in Sect.5] how transient impact is calibrated additively to log-prices, that means multiplicatively in prices. See also the comparison in   \cite[Example~5.5]{BechererBilarevFrentrup2016-deterministic-liquidation} for arguments in favour of impact to be multiplicative if combined with multiplicative price dynamics of Black-Scholes-type. A large strand of literature
investigates (linear-)quadratic control problems in this realm (see e.g.\ \cite{BankSonerVoss16,AckermannEtal22OW}) which are different to the superhedging and -pricing problem.
An undesirable property of additive impact (\refeq{eqn:OWadditive}) in this context is that it can lead to negative asset prices $s$ for the underlying. 
It is plausible that trading a quantity of stocks, that is a fraction of company ownership, should have a relative (hence multiplicative) effect on the price.
Indeed, already Bertsimas and Lo  \cite[Sect.3]{BertsimasLo98} have argued that relative (percentage) 
price impact which is proportional to the traded number of stocks (i.e.\ additive impact with respect to log-price in first order approximation) is more plausible than absolute price impact, and they cite empirical evidence.

A simple way to obtain a multiplicative impact variant is by a log-linear interpretation of the additive Obizhaeva-Wang model \eqref{eqn:OWadditive}, simply by taking $s=\log S$, $\bar s =\log \baseS$ to be {log-prices} instead of nominal prices $S,\baseS$ (affected, respective fundamental). Then impact on $S=\baseS \exp (\lambda Y)$
is multiplicate and log-linear, with the resilience and the (log-)price impact functions from \eqref{eqn:OWadditive} being linear. 
This is the basic log-linear example  (see \cref{example:loglinear}) which is covered by and motivating our transient multiplicative impact model, with unaffected price process $\baseS$ for the underlying asset being of Black-Scholes-Merton type. 
Our analysis moreover admits for non-linear 
and non-parametric
resilience and price impact functions $h$ and $f$ (in \eqref{eq:impact process},
\eqref{eq:price process}).
The model is a multiplicative variant of the (non-linear) additive impact model from \cite{PredoiuShaikhetShreve11}, where  price impact can be interpreted in terms of a limit order book shape that is static with respect to relative price perturbations  with $Y$ being their volume effect process
(see \cite[Sect.2.1]{BechererBilarevFrentrup2016-deterministic-liquidation}).

The contributions of the present paper are threefold: (1) We solve the superhedging problem under transient price impact
which is multiplicative, instead of additive. (2) Our results account for settlement specifications imposed at maturity which require analysis in \emph{liquidation values} instead of book (paper) value, such that physical units of the underlying risky asset and cash matter at maturity (and as well at initial) time, i.e. terminal (initial) price impact cannot be treated as null. Following terminology by \cite{BLZ16b} 
(cf.\ \cref{rem:extension-covered-options}), this means
we solve the hedging problem for non-covered instead of covered options. 
(3) In this realm, the model we study is basically complete,  
with transient price impact being 
the only digression from the friction-less Black-Scholes model assumptions (for $\baseS$), and it yields non-trivial extensions to the classical no-arbitrage pricing and hedging, while avoiding paradoxical effects from illiquidity modelling as mentioned in \cite{CetinSonerTouzi10}, 
without additional further frictions (like transaction costs, or constraints on trading strategies to be ``small'').
In particular, the large trader has neither the ability to ``manipulate'' \cite{Jarrow92,BankBaum04}
the market to achieve unreasonable profits (see \cref{remark:profitable-round-trips}  and subsequent remarks) nor to sidestep liquidity costs entirely and trade in effect like a small trader by exploiting modelling artefacts that occur in lack of sensible continuity properties (cf.\ \cite[Sect.3]{BechererBilarevFrentrup-model-properties}). 

We formulate the superhedging problem 
as a stochastic target problem and 
prove a Dynamic Programming Principle (DPP) along 
reduced coordinates
for \emph{the effective price and impact processes}, 
which represent the price and impact levels that would prevail if the large trader were to unwind her (long or short) position in the underlying risky asset immediately. Along the reduced coordinates, the DPP  provides a way to compare at stopping times the \emph{instantaneous liquidation wealth} and the (minimal) superhedging price, what permits to characterize the  superhedging price as the viscosity solution 
to  a non-linear pricing PDE, which is  a semi-linear extension of the Black-Scholes equation, whose non-linearity involves the (non-parametric) price impact and the resilience functions as well.
If the PDE has a sufficiently regular solution, it yields an optimal strategy which is even \emph{replicating} the option payoff in the required settlement units.
This strategy incorporates the transient nature of impact in that it depends on the effective level of impact.
Our analysis is also motivated by analytical tractability. It shows how effects from transience of price impact
arise in a basically complete model
without other additional frictions from transaction costs or constraints,
with scope for results beyond the ones of the present paper as outlined in \cref{sec:permanent impact}.

While there is a large literature on optimal execution and portfolio optimization problems under transient price impact, mostly for price impact being additive but also for multiplicative impact (see \cite{ObizhaevaWang13,AlfonsiSchiedSlynko12,BussetiLillo12} or \cite{GuoZervos13,BechererBilarevFrentrup2016-stochastic-resilience,BechererBilarevFrentrup-model-properties} and references therein),
the literature on superhedging (or perfect hedging, i.e.\ replication) under price impact, as stated above, is mostly treating permanent and purely instantanuous price impact (transaction costs, possibly non-linear) or a combination of the two \cite{Frey1998,SchoenbucherWilmott2000,BankBaum04,CetinSonerTouzi10,FreyPolte11}, with impact often being taken in multiplicative form. For the implications of option settlement specifications on hedging, only few papers admit for price impact also at maturity. Clearly, it requires some relevant non-zero price impact at maturity to obtain differences between  settlement specifications for options in physical or cash units, as in \cite{BLZ16a}. Yet, most articles \cite{Frey1998,
	CetinSonerTouzi10,
	FreyPolte11} treat another hedging problem which is not posed 
in terms of hedgable units of assets but instead 
in terms of book (that is ``paper'') value, with price impact at maturity (and possibly initiation) in the analysis being effectively taken to be zero. Such relates to a  different hedging problem for ``covered'' options 
(see 
\cref{rem:extension-covered-options}). 
A major difference to the work by Bouchard et al.\ \cite{BLZ16a}, that offered a fresh view to the hedging problem and inspired ours, is that the analysis in \cite{BLZ16a} is for permanent and additive impact. In contrast, our hedging results show non-trivial effects from transience of price impact, that is multiplicative. 
Whereas the basic example to \cite{BLZ16a} is the Bachelier model with additive impact, our basic example is a Black-Scholes-type model with transient multiplicative price impact (see \cref{example:loglinear} and \cref{rem:BLZimpactcomparison}) that is the
log-linear variant of the model by \cite{ObizhaevaWang13}. More detailed comparisons are provided thoughout the paper.

The paper is organised as follows. Sections~\ref{sec:impact model}-\ref{sect:options in illiquid markets}
introduce the model of transient multiplicative price impact and formulate the hedging problem. Effective coordinates for dynamic programming in ``liquidation magnitudes'' are explained in \cref{sec:stoch target problem}.
\cref{sec:The pricing PDEs} identifies hedging prices by viscosity solutions to semilinear PDEs (possibly degenerate, with delta constraints), with technical proofs deferred to \cref{subsec:Proof of viscosity solution properties}.
Results are illustrated by numerical examples in \cref{sec:Numerics}. Finally,
\cref{sec:permanent impact} extends results to combined transient and permanent impact, points out further possible extensions to cross-impact with multiple assets, and comments on related results to the different hedging problem for covered options.


\section{A multiplicative transient price impact model}
\label{sec:impact model}
This section describes the model 
for this paper. An extension with additional permanent impact is described in \cref{sec:permanent impact}. 
Let  $(\Omega,\F,\PP)$ be a complete probability space with countably generated $\F$,  a filtration $\FF = (\F_t)_{t\geq 0}$ satisfying the usual conditions  and an $\FF$-Brownian motion $W$. 
We take semimartingales to have c\`{a}dl\`{a}g paths, 
$\RR_+ = (0, +\infty)$ and $\inf \emptyset = +\infty$.

The \emph{unaffected} price process $\baseS$ of the underlying risky asset evolves, if the large trader (her) is inactive, according to the stochastic differential equation
\begin{equation}
	\diff \baseS_t = \baseS_t (\mu_t \diff t + \sigma \diff W_t), \quad \baseS_0 \in \RR_+,
\end{equation}
with constant $\sigma>0 $ and  bounded progressive process $\mu$. 
The c\`{a}dl\`{a}g adapted process $\assetsProcess$ denotes the evolution of her holdings (in units of shares) in the risky asset, say a stock, which is the underlying for the derivative contingent claim in the hedging problem. The market \emph{impact} process $Y = Y^{\assetsProcess}$ is defined pathwise on the Skorohod space of c\`{a}dl\`{a}g paths, by  
\begin{equation}\label{eq:impact process}
	\diff Y^\assetsProcess_t = -h(Y^\assetsProcess_t)\diff t + \diff \assetsProcess_t, \quad Y_{0-} = y\in \RR,
\end{equation} 
for a \emph{resilience function} $h:\RR\to \RR$ which is Lipschitz continuous function with $\sgn(x)h(x)\geq 0$, as in \cite{BechererBilarevFrentrup2016-deterministic-liquidation,BechererBilarevFrentrup-model-properties}. 
When the large trader  trades dynamically according to strategy $\assetsProcess$, the risky asset price observed on the market,
being the marginal price at which an additional infinitesimal quantity could be traded,
is
\begin{equation}\label{eq:price process}
	S_t:=S^{\assetsProcess}_t := f(Y^\assetsProcess_t) \baseS_t, \quad t\geq 0,
\end{equation} 
where the price \emph{impact function} $f:\RR\to \RR_+$ is increasing and in $C^1$ with $f(0) = 1$.
In particular,  $\lambda:= f'/f$ is a non-negative and locally integrable  $C^0$ function, satisfying 
\begin{equation}\label{eq:impact fct}
	f(x) = \exp\left(\int_0^x \lambda(u)\diff u\right), \quad x\in \RR.
\end{equation}

\revremove{By the monotonicity of $f$, the price impact from her trades is adverse to the large trader.
	During periods where the large trader is inactive, the 
	impact process $Y$ recovers towards its neutral state $0$, so that the relative price  impact $S/\baseS=f(Y)$  w.r.t.~the unaffected (fundamental) price $\baseS$ is persistent but lessens over time, rendering the 
	impact as \emph{transient}.
}

\begin{example}\label{example:loglinear}
	The basic example is transient proportional price impact with unaffected prices $\bar S$ being geometric Brownian motion, as in the Black-Scholes model with $\mu\in \RR$ being constant, for resilience $h(y)=\beta y $   and  log-price impact $\log f(y)=\lambda y$ being linear functions with constants $\beta,\lambda\in \RR_{+}$. 
	Then multiplicative price impact is proportional 
	to the number of shares $\Delta \Theta_t= \Theta_t- \Theta_{t-} $ traded at time $t$, that is linear
	in log-prices 
	\begin{equation}
		\log(S_{t+\delta t})-\log(S_{t-})=\lambda \Delta \Theta_t 
	\end{equation}
	with exponential decay $\log(S_{t+\delta })=\log(S_{t}) \exp(-\beta \delta )$ over time, when there are no further trades within time period $(t,t+\delta]$). For such linear choice of $h$ and $\log f$, log asset prices $\log S_t$ under multiplicative impact evolve like nominal asset prices
	in the seminal model by \cite{ObizhaevaWang13} for  additive transient price impact, as decribed in equations \eqref{eqn:OWadditive}.
	
	Our setting also admits for \emph{resilience rate} $\beta:=0$ (hence $h$) to be zero, what makes price impact permanent (cf.~\cref{sec:permanent impact}) and log-price impact $\log(S_t/S_{0-})$ linear in $Y_t-Y_{0-}=\Theta_t-\Theta_{0-}$.
\end{example}

Next, we specify the large trader's proceeds (negative expenses) $L$, which are the variations of her cash account to fund the dynamic holdings $\assetsProcess$ in the risky asset. For simplicity, we assume zero interest and a riskless asset with 
constant price $1$ as cash, i.e. prices are discounted 
in units of this numeraire asset.
For continuous  strategies $\assetsProcess$ of finite variation,  
\begin{equation}\label{eq:proceeds for fv}
	L(\assetsProcess) = -\int_0^{\cdot} S^{\assetsProcess}\diff \assetsProcess
\end{equation}
are the proceeds. And
there is a unique continuous  extension of the functional $\Theta\mapsto L(\Theta)$ in \eqref{eq:proceeds for fv} to general 
(bounded) 
semimartingale strategies $\assetsProcess$, that is given   
by 
\begin{equation}\label{eq:proceeds for semimart}
	L(\assetsProcess) := \int_0^\cdot F(Y^\assetsProcess_t)\diff \baseS_t - \int_0^\cdot \baseS_t (fh)(Y^\assetsProcess_t)\diff t  - (\baseS F(Y^\assetsProcess) - \baseS_0 F(Y^\assetsProcess_{0-})),
\end{equation}
as shown in \cite[Theorem~3.8]{BechererBilarevFrentrup-model-properties}, with antiderivative
\begin{equation}\label{eq:def of F}
	F(x):= \int_0^x f(u)\diff u, \quad x\in \RR.
\end{equation}
More precisely,  every (c\`{a}dl\`{a}g) semimartingale can be approximated (in probability) in the  Skorokhod space $D([0,T])$ of c\`{a}dl\`{a}g paths with the Skorokhod $M_1$-topology (cf.~\cite[Sect.3.1]{BechererBilarevFrentrup-model-properties}) by a sequence of continuous processes of finite variation;
And for semimartingales $\assetsProcess^n \xrightarrow{\PP} \assetsProcess$ in $(D([0,T]), M_1)$ converging to a semimartingale $\assetsProcess$, then $L(\assetsProcess^n)\xrightarrow{\PP} L(\assetsProcess)$ in $(D([0,T]), M_1)$.
To define $L$ by \eqref{eq:proceeds for semimart} is thus natural as the continuous extension of $L$ from \eqref{eq:proceeds for fv} to all semimartingales.

For other possible applications, it is useful to note 
that, more generally, there is a unique continuous extension even beyond semimartingale strategies, see \cite[Sect.3]{BechererBilarevFrentrup-model-properties}; cf.\   \cite{HorstKivman21,AckermannEtal22OW} for use of related continuity arguments in different applications. But for our hedging problem, semimartingale strategies will suffice, see \eqref{eq:admissible strategies}. 
In section \cref{sec:stoch target problem}, the superhedging problem of \cref{def:hedging of non-covered options} and the superhedging price 
\eqref{eq:def of superhedging price} are defined with respect to a particular set of admissible strategies 
(see \eqref{eq:admissible strategies}). The form of this set (being as in \cite{BLZ16a}) plays a technical role for proofs for the geometric dynamical programming (of \cref{thm:GDPP}). One may ask, to which extent the particular choice of this set affects the superhedging price. We will see that in base cases, the superhedging prices $w$ basically recovers impact-friction-less Black-Scholes prices: See \cref{cor:BSpde} and	\cref{rmk:permanent impact and constant lambda}, and likewise in \cite{BLZ16a} (with respect to the Bachelier model). This indicates, \emph{that} the superhedging price $w$ defined in \eqref{eq:def of superhedging price} should be robust in that is does not depend on particularities of the said set. The almost-sure uniform approximation result by continuous finite variation strategies in \cite[Proposition~3.12]{BechererBilarevFrentrup-model-properties} explains, \emph{why} such robustness indeed holds in the superhedging problem with the notion of liquidation wealth processes, which are given by the unique continuous extension from more elementary to more general trading strategies (than those in \eqref{eq:admissible strategies}).

The proceeds from a block trade of selling $\Delta \assetsProcess_t$ shares at time $t$ are 
\begin{equation}\label{eq:blocktrade}
	-\baseS_t \int_0^{\Delta \assetsProcess_t}f(Y^\assetsProcess_{t-} + x)\diff x,
\end{equation}
showing that the price per share that the large trader pays (resp. obtains) for a block buy (resp. sell) order is between the price before the trade $f(Y^{\assetsProcess}_{t-})\baseS_t$ and the price after the trade $f(Y^{\assetsProcess}_{t})\baseS_t$.
The form of  proceeds and price impact from block trades can be interpreted from the perspective of a  
latent 
limit order book, where a block trade
is 
executed against available orders in the order book for prices between $f(Y^{\assetsProcess}_{t-})\baseS_t$ and $f(Y^{\assetsProcess}_{t-}+\Delta \assetsProcess_t)\baseS_t$, see  \cite[Section~2.1]{BechererBilarevFrentrup2016-deterministic-liquidation}, and 
$Y$ can be understood as a volume effect process in spirit of \cite{PredoiuShaikhetShreve11}.

For a self-financing strategy $(\assetbeta, \assetsProcess)$, in which the dynamic holdings  in cash (the riskless asset, savings account) and in stock (the risky asset)
evolve as  $\assetbeta$ and  $\assetsProcess$, 
the self-financing condition is
\[
\assetbeta = \assetbeta_{0-} + L(\assetsProcess).
\]
In order to define a wealth dynamics for the large trader's strategy, it remains to specify the value of the risky asset position in the portfolio in a suitable way.
If the large trader were forced to liquidate her position of $\assetsProcess_t$ stocks immediately by a single block trade at market prices, her \emph{liquidation wealth}  $V^{\text{liq}}_t = V^{\text{liq}}_t(\assetsProcess)$ at  time $t\ge 0$ 
is 
\begin{equation}\label{eq:wealth process}
	V^{\text{liq}}_t(\assetsProcess) := \assetbeta_t + \baseS_t\int_{0}^{\assetsProcess_t} f(Y^\assetsProcess_t - x)\diff x
	= \assetbeta_{0-} + L(\assetsProcess)_t +\baseS_t\int_{0}^{\assetsProcess_t} f(Y^\assetsProcess_t - x)\diff x.
\end{equation}
This 
wealth 
process is mathematically conveniently tractable, evolving continuously with
\begin{equation}\label{eq:inst liq value}
	\diff V^{\text{liq}}_t = (F(Y_{t-}) - F(Y_{t-}-\assetsProcess_{t-}))\diff \baseS_t - \baseS_t (f(Y_{t-}) - f(Y_{t-}-\assetsProcess_{t-}))h(Y_t)\diff t,
\end{equation}
and $V^{\text{liq}}_0=B_0-$, 
and it inherits from the proceeds \eqref{eq:proceeds for semimart} the stable continuous dependence properties (on $\assetsProcess$), mentioned above.
The notion of liquidation wealth $V^{\text{liq}}(\assetsProcess)$ is relevant for the hedging application of \cref{sect:options in illiquid markets} and it is different from the so-called book wealth process
\begin{equation}\label{eq:def book value}
	V^{\text{book}}(\assetsProcess) := \assetbeta + \assetsProcess S =\assetbeta_{0-} + L(\assetsProcess) + \assetsProcess S,
\end{equation}
in which risky assets are evaluated at the current marginal market price $S$. Because of price impact (monotonicity of $f$, positivity of $f,\baseS,S$), clearly 
$V^{\text{liq}}_t\le V^{\text{book}}_t$. 
In the terminology of Jarrow~\cite[cf. Sect.IV]{Jarrow92}, $V^{\text{liq}}$ is real wealth whereas $V^{\text{book}}$ is paper wealth. Recently, Kolm and Webster \cite{KolmWebster23} have given theoretical and practical reasons that accounting for value (respectively P\&L, that means changes in value) of a risky asset position based on current market prices $S$ () can be misleading and needs to be adjusted for price impact; in their terminology $V^{\text{liq}}$ corresponds to fundamental wealth whereas $V^{\text{book}}$ is accounting wealth, also referred to as mark-to-market wealth.

One obtains
from \eqref{eq:inst liq value} 
absence of arbitrage
within the following set of admissible
strategies
\begin{align*}
	\mathcal{A}^{\text{NA}} := \big\{(\assetsProcess_t)_{t\geq 0} \mid {}&\text{bounded semimartingale,  with $\assetsProcess_{0-} = 0$}
	\\  		&\text{and $\assetsProcess_t = 0$ on $t\in[ T,\infty)$ for some $T\in (0,\infty)$} 
	\big\}.
\end{align*}

\begin{proposition}\label{thm:absence of arbitrage}
	The market is free of arbitrage up to any finite time horizon $T \in [0,\infty)$ in the sense that there exists no 
	$\assetsProcess \in \mathcal{A}^{\text{NA}}$ with $\assetsProcess_t=0$ on $t\in[T,\infty)$ such that for the self-financing strategy $(\assetbeta, \assetsProcess)$ with $V^{\text{liq}}_{0-}:=\assetbeta_{0-} = 0$ we have $\PP[ V^{\text{liq}}_T \geq 0 ] = 1$ and  $\PP[V^{\text{liq}}_T > 0 ] > 0$.
	Moreover, for any such $(\assetbeta, \assetsProcess)$ there exists a probability measure $\QQ^\assetsProcess\sim \PP$ equivalent to $\PP$ (on $\F_T$) 
	such that $ V^{\text{liq}}$ is 	
	a 
	$\QQ^\assetsProcess$-martingale. 
\end{proposition}

In the terminology of \cite[Sect.IV, eqn.(13)]{Jarrow92}, the no-arbitrage result of \cref{thm:absence of arbitrage} states that there exist no \emph{market manipulation trading strategies}. Note that, 
in contrast, there is no reason to expect no-arbitrage in terms of book wealth $V^{\text{book}}$; there are simple counterexamples, see \cref{ex:na-realvspaper} for implications on (super-)hedging prices.

\begin{remark}\label{remark:profitable-round-trips} 	
	In the seminal article by \cite{HubermanStanzl04}, a notion of \emph{no profitable round-trips} (stronger than no-arbitrage) is defined, which (in our notation) requires that there exists no (self-financing) strategy $(B_{0-}, \Theta)$ as in \cref{thm:absence of arbitrage} with $V^{\text{liq}}_{0}= 0$
	and $E[V^{\text{liq}}_{T}]>0$. This means that there is no (self-financing) strategy from zero initial holdings (i.e. $V^{\text{liq}}_{0}=0$), that achieves a terminal liquidation wealth $E[V^{\text{liq}}_{T}]>0$ which is positive in expectation, within a compact time interval  $[0,T]$. Recall that, by definition, liquidation wealth $V^{\text{liq}}$ is the value of an only-cash position held after all stock holdings are liquidated.  
	
	A much cited result from \cite{HubermanStanzl04} states that 
	price impact needs to be linear to exclude profitable round trips.
	This is not in conflict with our modelling, as the proof in \cite{HubermanStanzl04} relies, of course, on some assumptions. They include permanent and additive impact. 
	For comparison, under multiplicative permanent price impact, a linear log-price impact function $\log f$ is sufficient to conclude that $V^{\text{liq}}$ is a martingale under $\PP$ (by \eqref{eqn:log-linear-permanent-impact-wealth}
	in \cref{rmk:permanent impact and constant lambda}), if $\baseS$ is a $\PP$-martingale (e.g.\  geometric Brownian motion, like in the Black-Scholes model under risk-neutral measure). This implies  $\EE[V^{\mathrm{liq}}_{T}]=\EE[ V^{\mathrm{liq}}_{0-}]$, excluding profitable round-trips.
\end{remark}

\begin{example}\label{ex:na-realvspaper}
	To explain why
	absence of arbitrage  in terms of liquidation values as  above (which is, notably, not available in book values) is a relevant property for almost-sure hedging problems, 
	let us compare hedging in liquidation (``real'') and book (``paper'') values in a simple example within the basic setting of  \cref{example:loglinear}: Consider the European option whose derivative payout (in cash, say) at maturity $T$ is $h(S_T):=S_T(1- \exp(-\lambda))$, being strictly positive almost surely. This payoff cannot be super-replicated in terms of 
	liquidation value from non-positive initial wealth, i.e. there exists no 
	$\assetsProcess \in \mathcal{A}^{\text{NA}}$ with 
	$V^{\text{liq}}_{0-}\le 0$ and $V^{\text{liq}}_T \geq h(S_T)$.
	In contrast, in book values, a basic computation shows that the self-financing strategy $(\assetsProcess,B)$ with 
	$\mathcal{A}^{\text{NA}}\ni \assetsProcess:= 1_{[T,\infty)}$ and
	$\assetbeta_{0-} := 0$ satisfies 
	$V^{\text{book}}_{0-}=V^{\text{liq}}_{0-}=0$ but $V^{\text{book}}_{T}=h(S_T)>0$ (whereas $V^{\text{liq}}_{T}=0$). This illustrates a major distinction between (super-)hedging problems posed in liquidation values and those posed in book values, see \cref{rem:extension-covered-options}.
\end{example}

\begin{proof}(of \cref{thm:absence of arbitrage})
	The idea of proof is as in 
	\cite[Sect.4]{BechererBilarevFrentrup-model-properties}, where it was additionally required  for admissible strategies that $V^{\text{liq}}$ is bounded from below.  The latter condition however 
	can be omitted
	in the present setup of bounded strategies. To see this, observe that for any $ \assetsProcess\in \mathcal{A}^{\text{NA}}$  there exists an equivalent measure $\QQ^\assetsProcess\sim \PP$ (on $\F_T$),  
	constructed as in \cite[proof~of~Thm.4.3]{BechererBilarevFrentrup-model-properties}, 
	under which the  wealth process $ V^{\text{liq}}$ is  
	a 
	martingale.
\end{proof}

\begin{remark}\label{rem:BLZimpactcomparison}
	To highlight key differences to \cite{BLZ16a}, let us explain in 
	detail why the basic (log)-linear example for  our setup is the Black-Scholes model for $\baseS$ (geometric Brownian motion) with multiplicative (proportional, relative) price impact (see \cref{example:loglinear}), whereas for  \cite{BLZ16a} the basic example is the Bachelier model (additive Brownian motion) with additive linear price impact.
	Note first that \cite[see equation (2.1)]{BLZ16a} study a general model where price impact is permanent and
	additive, in the sense that (using notation as in our paper)  resilience  $h=0$ is zero, thus $Y=\Theta$ for $Y_{0-}=\Theta_{0-}:=0$, and the stock price  after a small (infinitesimal) 
	trade of size $\delta$ becomes $s(\theta+\delta )\approx s(\theta)+\delta \mathfrak{f}(s(\theta))$ 
	where $\mathfrak{f}:\mathbb{R}\to (0,\infty)$ is a smooth function of the current stock price  $s(\theta)$ which prevails if the large trader holds $\theta$ stocks  just before the trade. That means,  more precisely, 
	$\frac{d }{d\theta}s(\theta)=\mathfrak{f}(s(\theta))$. 
	For comparison, it is instructive to pretend, just formally, that one could choose a `multiplicative' form $\mathfrak{f}(x):=\lambda(x) x$. With $\bar{s}:=s(0)$ one then would get
	$s(\theta)=\left(\exp(\int_0^\theta \lambda(x)dx)\right) \bar{s}$, being reminiscent to (\refeq{eq:price process})--(\refeq{eq:impact fct}),
	and taking  $\lambda$ to be constant would give $s(\theta)=\exp(\lambda \theta) \bar{s}$, what is the permanent impact 
	variant of the basic
	case for multiplicative (transient) impact that is studied in \cref{subsec:The case of constant lambda}. 
	However, a choice like $\mathfrak{f}(x)= \lambda x$ in linear (multiplicative) form with $\lambda> 0$ does not fit with assumptions (H1) and (H2) in \cite{BLZ16a}: Neither is $x\mapsto \lambda x$  (strictly) positive on $\mathbb{R}$,
	nor is $x\mapsto \exp(\lambda x)$ a surjective function from $\mathbb{R} \to \mathbb{R}$.  Observe that asset prices
	in \cite{BLZ16a} take values $x$ in $\mathbb{R}$ (instead of $(0,\infty)$). The instructive basic example for their setting is
	the case of fixed (constant) impact with $\mathfrak{f}(x):=\lambda>0$ where $s(\theta)=\bar{s}+\lambda \theta$, 
	and with  the unaffected asset price $\bar{s}$ evolving as in the Bachelier model (say), see~\cite[Section~3.4]{BLZ16a} with additive permanent price impact. In contrast, the basic  example to our setup is transient proportional impact (being additive and linear in terms of log-prices) with respect to a Black-Scholes-type geometric Brownian motion for $\bar S$, see \cref{example:loglinear}. 
	
\end{remark}

\section{
	Hedging under transient price impact
}
\label{sect:options in illiquid markets}
We solve in Sections~\ref{sect:options in illiquid markets}-\ref{sec:Numerics} the common problem of dynamic hedging,  
where the issuer who wants to hedge the option receives at time $t=0$ the option premium in cash. 
In an illiquid market setting with price impact, it is relevant to distinguish between \emph{cash settlement} and \emph{physical settlement} of an option payoff because, in contrast to frictionless models with unlimited liquidity, moving funds 
between 
the bank account and the risky asset account 
not only induces trading costs from price impact but also affects the price evolution of the underlying, what induces feedback effects \cite{Frey1998,SchoenbucherWilmott2000}. Depending on the option's settlement specifications, a terminal block trade at maturity could affect an option's payoff in different ways (see \cref{sec:Numerics}).  
We  consider contingent claims of the following type.
\begin{definition}
	A European option with maturity $T\ge 0 $ is specified by a measurable map
	$$ (g_0,g_1): (s,y) \in \RR_{+}\times \RR\mapsto (g_0(s,y), g_1(s,y))  \in   \RR\times \RR$$  
	representing the payoff, with cash-settlement part $g_0$ and physical-delivery part $g_1$ at maturity. It entitles its holder to receive  $g_0(S_T, Y_T)$ in cash and $g_1(S_T, Y_T)$ in units of the underlying risky asset, when $(S_T, Y_T)$ is the risky asset price and the level of market impact at maturity.
\end{definition}

Henceforth, $T$ is a fixed maturity time. The optimization task for the seller, that is the issuer, of the option with payoff $(g_0, g_1)$ is to do dynamic hedging at minimal cost to avoid potential losses from her obligation to deliver the payoff at maturity.

Among her admissible trading strategies $\Gamma$ (to be specified precisely in Section~\ref{subsec:targetformu}), she is going to look for the 
cheapest strategies to super-replicate the option's payout in the following sense. 
\begin{definition}[Hedging of an option]\label{def:hedging of non-covered options}
	A superhedging strategy is a self-financing strategy $(\assetbeta, \assetsProcess)$ with $\assetsProcess\in \Gamma$, $\assetsProcess_{0-} = 0$, and  
	\[
	\assetbeta_T \geq g_0 (S_T, Y_T) \quad \text{and}\quad \assetsProcess_T = g_1(S_T, Y_T).
	\]
\end{definition}
We emphasize that a hedging strategy has to deliver  the physical component $g_1(S_T, Y_T)$ at maturity exactly, and that any further (long or short) position in the underlying has to be unwound before options are settled at the resulting price $S_T$  and impact level  $Y_T$. 
In particular, a hedging strategy for a payoff with pure cash delivery part is a so-called round trip, i.e.~it begins and ends with zero shares in the underlying, while the hedging strategy for a payoff with non-trivial physical delivery part should be such that the amount of risky assets held at maturity will meet exactly the physical delivery requirement. Thus, hedging strategies for European contingent claims with physical delivery can be different from those with pure cash delivery part, and we will see, that their respective prices can also differ.

The (minimal) superhedging price of 
an option with payoff $(g_0, g_1)$
is the minimal (infimum of) initial capital $\assetbeta_{0-}$ for which such  a superhedging strategy $(\assetbeta, \assetsProcess)$ exists.
Note that by the impact process $Y$,  the hedging strategy $\assetsProcess$ clearly affects the volatility of the price process $S$ underlying the option payout, because price impact in \eqref{eq:price process} is multiplicative.

Options with pure cash settlement are described by $g_1 = 0$.  Every (reasonable) option could be represented by a payoff with pure cash settlement. Indeed, if the $\Gamma$ set is stable under adding additional jump at maturity time, meaning that $\assetsProcess\in \Gamma$ implies that $\assetsProcess + \Delta\indicator_{\{T\}}\in \Gamma$ for every $\F_T$-measurable $\Delta$, then any European option can be represented by an option with pure cash settlement. To see this for an option with payoff $(g_0, g_1)$, let for $(s,y)\in \RR_+\times \RR$
\begin{equation}\label{eq:pure cash equivalent non-cov optn}
	H(s,y):= \inf \big\{g_0\big(s\tfrac{f(y+\theta)}{f(y)}, y+\theta \big) + s \tfrac{F(y+\theta) - F(y)}{f(y)} \mid \theta = g_1\big(s\tfrac{f(y+\theta)}{f(y)}, y+\theta\big)\big\}.
\end{equation}
The value $H(s,y)$ is the minimal 
amount of cash (riskless assets) needed  to hedge the payoff $(g_0, g_1)$ with a single (instant) block trade at maturity, when just before that trade (at time $T-$) the level of impact is $y$ and there are no holdings in the risky asset whose price is $s$. Indeed, a block trade of size $\theta$ will result in the new price $\tilde{s}= sf(y+\theta)/f(y)$ and impact $\tilde{y} = y+\theta$,  it will incur the cost $s(F(y+\theta) - F(y))/f(y)$.  Thus, it will hedge the claim $(g_0, g_1)$ if  $\theta = g_1(\tilde{s}, \tilde{y})$ and we have enough capital to pay for the block trade and to cover the cash-delivery part that after the block trade equals $g_0\big(\tilde{s}, \tilde{y})$, see \cref{def:hedging of non-covered options}.

\begin{example}
	1. A cash-settled European call option with strike $K$  is specified by the payoff $(g_0(s,y),g_1(s,y))=((s - K)^{+}, 0)$.
	
	2. In comparison, a European call option with strike $K$ and physical settlement has the payoff $(-K\indicator_{\{s\geq K\}}, \indicator_{\{s\geq K\}})$. Although the payoff profile $(g_0,g_1)$ does not directly depend on the level of impact $y$, the equivalent pure cash settlement profile $H$ from \eqref{eq:pure cash equivalent non-cov optn} typically will depend on it, if the function $\lambda$ is not constant. Indeed, the effect on the relative price change $f(y+\theta)/f(y)$ from a block trade $\theta$ can depend on the level $y$ of impact before the trade in general, unless $f(x) = \exp(\lambda x)$ for  $\lambda$ being constant (linear log-price impact). 
	
\end{example}

\begin{remark}\label{rmk:permanent impact and constant lambda}
	We discuss an example to show how the hedging problem for the large trader could be related to hedging in a market with perfect liquidity 
	but with portfolio constraints,  if $F$ from \eqref{eq:def of F} is not surjective onto $\RR$.
	In particular, in this case our market model will not be complete in the sense that not every contingent claim can be perfectly replicated. A prototypical example is the special case of  
	purely permanent impact, i.e. $h\equiv 0$, with constant $\lambda$ and log-linear price impact $\log f(x)= \lambda x$ (as in \cref{example:loglinear}), and an option whose payoff $(H,0)$ specifies settlement in cash only. Hence, we are in the setup of \cite{BankBaum04} with the smooth family of semimartingales $P(x,t) := \exp(\lambda x) \baseS_t$. If $Y_{0-} = 0$ and $\lambda = 1$, \eqref{eq:inst liq value} takes the form 
	\begin{equation}\label{eqn:log-linear-permanent-impact-wealth}
		\diff V^{\text{liq}}_t = (\exp( \assetsProcess_{t}) - 1) \diff \baseS_t.
	\end{equation}
	By the conditions from \cref{def:hedging of non-covered options}, any hedging strategy $\assetsProcess$ satisfies $\assetsProcess_T = 0$, and hence at maturity $S_T = \baseS_T$ and $Y_T = Y_{0-} = 0$. Thus, the superreplication condition becomes $ V^{\text{liq}}_T(\assetsProcess) \geq H(\baseS_T, 0)$. This means that, after a reparametrization $\assetsProcess\mapsto \exp(\assetsProcess) - 1$ of strategies, the superreplication problem in this large investor model  becomes equivalent to the 
	problem in the respective frictionless model (for instance, from Black-Scholes) with price process $\baseS$ for a small investor and with constraints on \emph{the delta} (to be greater than -1 ), that is on the number of risky assets that a hedging strategy might hold. In particular, one should expect that in such situations (where $F$ is not invertible) the pricing equation should contain gradient constraints. Note that this is different from \cite{BankBaum04} because for this particular $f$ the crucial Assumption 5 there is violated, and also different from \cite{BLZ16a} because their assumption (H2) would not hold in this case. 
	
	In the presence of resilience for the market impact ($h\not \equiv 0$),  the situation becomes more complex, however, since the evolution of the price and impact processes depend on the entire history of the trading strategy, and thus a simplification as above is not applicable. But we  will see later in \cref{subsec:The case of constant lambda} that in the case $f = \exp(\lambda\cdot)$ a lower bound on the delta will also emerge naturally in order to make sense of the pricing equation.
\end{remark}

\section{Superhedging by geometric dynamic programming}
\label{sec:stoch target problem}
We formulate the superhedging problem 
as a stochastic target problem and prove a geometric Dynamic Programming Principle (DPP) 
for the control problem whose value function will be characterized.
Notably, it will show that a  DPP (\cref{thm:GDPP}) holds with respect to suitably chosen coordinates, which correspond to 
modified state processes describing the evolution of \textit{effective} price and impact levels that would  
result from an 
immediate unwinding
of the risky asset holdings by the large trader. 
With respect to these new effective coordinates, we will 
characterize the value function of the 
control problem  as a viscosity solution to a partial differential equation, cf.~\eqref{eq:pricing pde bounded f} and \eqref{eq:pde const lambda} in \cref{sec:The pricing PDEs}, that is the pricing PDE
generalizing the (frictionless)  Black-Scholes equation.

\subsection{Stochastic target formulation}\label{subsec:targetformu}

We consider strategies that take values in the constraint set $\K\subseteq \RR$, for one of the two cases    
\begin{align}
	&\K = [-K, +\infty)\text{  for some } K > 0, \text{ or} \label{eq:short-selling constraints}
	\\
	&\K = \RR. \label{eq:no delta constraints}
\end{align}
The short-selling constraints \eqref{eq:short-selling constraints} will be needed when $F$ is not surjective onto $\RR$, see \cref{rmk:permanent impact and constant lambda}, in which case we will consider in \cref{subsec:The case of constant lambda} $f(x) = \exp(\lambda x)$ for some $\lambda > 0$, while $\K = \RR$ will be in force when $f$ is bounded away from 0 and $+\infty$, meaning that the (relative) change of the price from a block trade cannot be arbitrarily big.

For the analysis, we need to allow for jumps in the admissible trading strategies in order to obtain a DPP, following \cite{BLZ16a}. For $k\in \NN$, let $\U_k$ denote the set of random $\{0,\ldots, k\}$-valued measures $\nu$ supported on $[-k, k]\times [0,T]$ that are adapted in the following sense: for every $A\in \mathcal{B}([-k,k])$, the process $t\mapsto \nu(A, [0,t])$ is adapted to the underlying filtration. Note that the elements of $\U_k$ have the representation
\[
\nu(A, [0,t]) = \sum_{i=0}^k \1_{\{(\delta_i, \tau_i) \in A\times [0,t]\}},
\]
where $0\leq \tau_1< \cdots < \tau_k \leq T$ are stopping times and $\delta_i$ is a $[-k,k]$-valued $\F_{\tau_i}$-random variable (might take values 0 as well). Consider also $\U := \bigcup_{k\geq 1} \U_k$.

The admissible trading strategies $\assetsProcess$ that we consider are bounded, take values in $\K$ and have the representation
\begin{equation}\label{eq:admissible strategies}
	\assetsProcess_t = \assetsProcess_{0-} + \int_0^t a_s\diff s + \int_0^t b_s \diff W_s + \int_0^t \int_\RR \delta \nu(\diff \delta, \diff s), \quad t\in [0,T],
\end{equation} 
in which $\assetsProcess_{0-}\in \mathcal{K}$,  $\nu \in \U$ and   $(a, b)\in \A:= \bigcup_{k\geq 1} \A_k$, where for $k\ge 1$ we define
\begin{align*}
	\A_k := \left\{ (a,b)\mid \text{$a$ and $b$ are predictable with } |a|\vee |b|\le k, \ \mathrm d t\otimes \mathrm d \PP\text{-a.e.} \right\}.
\end{align*}
In this sense, we identify the trading strategies by triplets $(a, b, \nu)\in \A\times\U$. For $k\in \NN$ set
$$\Gamma_k := \{(a,b,\nu)\in \A\times \U_k\ :\ \assetsProcess \hbox{ from }\eqref{eq:admissible strategies} \hbox{ takes values in }\mathcal{K}\cap [-k, k]\}$$
and let
$\Gamma := \bigcup_{k\geq 1}\Gamma_k$.
To reformulate the superhedging problem in our price impact model as a stochastic target problem, consider for $(t,z) = (t, s, y, \theta, v)\in [0,T]\times \RR_+\times \RR \times \K \times \RR$ and $\gamma\in \Gamma$ the 
\emph{ state process}
\[
(Z^{t, z, \gamma}_u)_{u\in [t,T]} = (S^{t, z, \gamma}_u, Y^{t, z, \gamma}_u, \assetsProcess^{t, z, \gamma}_u, V^{\text{liq},t, z, \gamma}_u)_{u\in [t,T]},
\]
where the  processes $S^{t, z, \gamma}, Y^{t, z, \gamma}, \assetsProcess^{t, z, \gamma}$ and  $V^{\text{liq},t, z, \gamma}$ correspond to the price, impact, risky asset position and instantaneous liquidation wealth 
processes on $[t,T]$ for the control $\assetsProcess^{t, z, \gamma}$ associated with $\gamma$ (from the  decomposition like \eqref{eq:admissible strategies} but on $[t, T]$), when started at time $t-$ at $s,y,\theta$ and  $v$, respectively. 

Following the discussion in \cref{sect:options in illiquid markets}, for an 
European option whose payoff in cash- and physical units at maturity $T$ is described by a measurable map $(s,y) \in \RR_{+}\times \RR\mapsto (g_0(s,y), g_1(s,y))$, 
$\gamma\in \Gamma$ is a dynamic superhedging strategy if its state process is a.s.\ at maturity $T$ within the set
\[
\mathfrak{G}:= \big\{ (s, y, \theta, v)\in \RR_+\times \RR \times\K \times \RR \ :\ \theta = g_1(s,y),\ v- s\big(F(y) - F(y-\theta)\big)/f(y) \geq g_0(s, y)\big\}
\] 
which we call the \emph{target set}. The superhedging strategies are
\[
\mathcal{G}(t,s,y,\theta,v) := \bigcup_{k\geq 1}\mathcal{G}_k(t,s,y,\theta,v)
\]
for $\theta$ denoting the initial position in the risky asset, and with 
\[
\mathcal{G}_k(t,s,y,\theta,v) := \{\gamma \in \Gamma_k\ :\ Z_T^{t, s,y,\theta,v, \gamma}\in \mathfrak{G}\}.
\]
Our aim is to derive the 
(minimal) 
superhedging price 
\begin{equation}\label{eq:def of superhedging price}
	w(t, s, y) := \inf_{k\geq 1}w_k(t, s, y),  \quad\hbox{where}\quad w_k(t, s,y) := \inf\{v\ :\ \mathcal{G}_k(t,s, y, 0, v) \neq \emptyset\},
\end{equation}
in the case where the hedger holds no assets of the underlying initially.
Let us note that the value function depends on the constraint set $\K$ (via the target set $\mathfrak{G}$). Note also that the set of admissible superhedging strategies (identified with $\mathcal{G}(t,s,y,0,v)$) is a subset of $\A^{\text{NA}}$, meaning that the 
superhedging price of a non-negative payoff $H$ (considered  as pure cash delivery equivalent of $(g_0,g_1)$ in \eqref{eq:pure cash equivalent non-cov optn} that is positive with non-zero probability, is strictly positive. 



\subsection{Effective coordinates and dynamic programming principle}\label{sec:effectice}

For stochastic target problems usually a form of the dynamic programming principle holds and plays a crucial role in deriving a PDE that characterizes the value function (in a viscosity sense). The aim of this section is to  provide a suitable DPP.

Let us first note that the formulation for the superhedging problem above looks not time-consistent, because in the definition  \eqref{eq:def of superhedging price} of the
superhedging price $w$ it is assumed that the initial position in risky assets is zero, whereas at later times it  typically will not be. To obtain a time-consistent formulation, the first naive  idea could be to make the risky asset position a new variable, that means to work with the function $\bar{w}$ defined on $[0,T]\times \RR_+ \times \RR \times \K$ by
\begin{equation}\label{eq:time consistent w}
	\hspace{-0.3cm}	\bar{w}(t, s, y,\theta) := \inf_{k\geq 1}\bar{w}_k(t, s, y,\theta) \ \hbox{with }\  \bar{w}_k(t, s,y,\theta) := \inf\{v\ :\ \mathcal{G}_k(t,s, y, \theta, v) \neq \emptyset\}.
\end{equation}
But the function  $\bar w(t,\cdot,\cdot,\cdot)$ would have to respect  a functional relation  along suitable orbits of the coordinates $(s,y,\theta)$ at any time $t$, because of the equations (\refeq{eq:price process}) and (\refeq{eq:blocktrade}),
namely 
$$\bar{w}(t,s,y,\theta)=\bar{w}\big(t,s f(y-\Delta)/f(y),y-\Delta,\theta-\Delta\big)+ (s/f(y))\int_0^{\Delta}f(y-x)dx,\quad \Delta\in \mathbb{R} .$$  
This suggests that one coordinate dimension is redundant and a `PDE on curves'  may be required to describe $\bar w$.  
Indeed, for our transient price impact problem we show how the state space can be reduced 
to make the analysis more transparent, by studying the problem in suitably reduced coordinates which can be interpreted as quantities (for price and impact $s,y$) at liquidation (of $\theta$)',  and 
with respect to which a DPP and a viscosity characterization is proven for the function $w$. Otherwise, we can follow arguments by \cite{BLZ16a}.%

To derive a dynamic programming principle for the function $w$, we want to compare it (evaluated at suitable coordinate processes) over time with the wealth process. 
Since by definition $w$ assumes zero initial risky asset holding, it is natural to consider the (fictitious) state processes that would prevail 
if the trader would be forced to liquidate her position in the risky asset immediately (with a block trade).
To this end, let
\begin{align}
	\calS(S_t,Y^\assetsProcess_t, \assetsProcess_t)  &:= \baseS_t f(Y^\assetsProcess_t- \assetsProcess_t) 
	= S_t f(Y^\assetsProcess_t- \assetsProcess_t)/f(Y^\assetsProcess_t)
	, \nonumber
	\\ \calY(Y^\assetsProcess_t, \assetsProcess_t) &:= Y^\assetsProcess_t -\assetsProcess_t.  \label{eqn:effective coord processes}
\end{align}
The process  $\calS(s,y,\theta)$ is interpreted as the price of the asset that would prevail after $\theta$ assets were liquidated, when $s$ and $y$ are the price of the risky asset and the market impact just before the trade, while  $\calY(y,\theta)$ would be the level of the market impact after this trade. 
In this sense, we refer to the processes $\calS(S_t,Y^\assetsProcess_t, \assetsProcess_t)$ and $\calY(Y^\assetsProcess_t,  \assetsProcess_t)$ as the \emph{effective price and impact processes}, respectively, 
for a self-financing trading strategy $\assetsProcess$. 
Observe that both processes are continuous, even though the trading strategy $\assetsProcess$ may have jumps.

For the dynamic programming principle in \cref{thm:GDPP}, 
we are going to compare the
\emph{liquidation wealth} $V^{\text{liq}}$, defined in \eqref
{eq:wealth process}, with the value function $w$ along evolutions of
the \emph{effective} price and effective impact processes $(\calS(S,Y^\assetsProcess, \assetsProcess), \calY(Y^\assetsProcess, \assetsProcess))$.

\begin{theorem}[Geometric DPP]\label{thm:GDPP}
	Fix $(t,s, y, v)\in [0,T]\times \RR_{+}\times \RR \times \RR$. 
	\begin{enumerate}[(i)]
		\item If $v > w(t, s, y)$, then there exists $\gamma \in \Gamma$ and $\theta\in \K$ such that 
		\[	
		V^{\text{liq}, t, z, \gamma }_\tau \geq w (\tau, \calS(S^{t,z,\gamma}_\tau, Y^{t,z,\gamma}_\tau, \assetsProcess^{t,z,\gamma}_\tau), Y^{t,z,\gamma}_\tau-\assetsProcess^{t,z,\gamma}_\tau)
		\]
		for all stopping times $\tau \geq t$, where $z = (\calS(s, y,-\theta), y+\theta, \theta, v)$.
		\item \label{part2thm}  Let $k\geq 1$. If $v< w_{2k+2}(t,s,y)$, then for every $\gamma\in \Gamma_k$, $\theta \in \K\cap [-k, k]$ and stopping time $\tau\geq t$ we have, with $z = (\calS(s,y, -\theta), y+\theta, \theta, v)$, that
		\[
		\PP \left[\  V^{\text{liq}, t, z, \gamma }_\tau > w_k(\tau, \calS(S^{t,z,\gamma}_\tau, Y^{t,z,\gamma}_\tau,  \assetsProcess^{t,z,\gamma}_\tau), Y^{t,z,\gamma}_\tau-\assetsProcess^{t,z,\gamma}_\tau)  \ \right] < 1.
		\]
		
	\end{enumerate}
\end{theorem}
\begin{proof}
	There are similarities and differences to  \cite[proof~of~Prop.3.3]{BLZ16a}, who treat the case for permanent additive impact, so we present the proof in full detail.  As explained in \cref{rmk:permanent impact and constant lambda}, the assumptions in \cite{BLZ16a} do not admit to cover multiplicative price impact. And transience of impact naturally requires a further dimension in the DDP.  
	The proof 
	uses general ideas on dynamic programming for stochastic target problems and geometric flows \cite{SonerTouzi02}.
	We emphasize that for showing the DPP, our proof develops mathematical arguments 
	in terms of effective coordinates and liquidation wealth
	$V^{\text{liq}}$,  
	what simplifies mathematical analysis and makes it more transparent. Such shows also in the ability for extensions, described in \cref{sec:permanent impact}.
	
	It is easy to see that for all $k\geq 2$ and  $(t,s,y,\theta)\in [0,T]\times \RR_+\times \RR \times (\K\cap [-k, k])$
	\begin{align}
		\bar{w}_{k}(t,s,y,\theta) &\geq  w_{k+1}(t, \calS(s,y,\theta), \calY(y,\theta)), \label{eq:ineq bar w and w in DPP 1}
		\\ w_{k-1}(t, \calS(s,y,\theta), \calY(y,\theta)) &\geq \bar{w}_{k}(t,s,y,\theta). \label{eq:ineq bar w and w in DPP 2}
	\end{align}
	Now suppose that $v > w(t,s,y)$. Then by definition of $w$ there exists $\theta\in \K$ and some $\gamma\in \mathcal{G}(t,z)$ for $z = (\calS(s,y, -\theta), y+\theta, \theta, v)$. As in \cite[proof~of~Thm.3.1, Step~1]{SonerTouzi02}, we have for all stopping times $\tau \geq t$  (the first part of) the DPP for $\bar{w}$:  $V_{\tau}^{\text{liq}, t, z, \gamma} \geq \bar{w}(\tau, S_\tau^{t,z,\gamma}, Y_\tau^{t,z,\gamma}, \assetsProcess_\tau^{t,z,\gamma})$. Then (i) follows from \eqref{eq:ineq bar w and w in DPP 1} by taking $k\rightarrow \infty$. 
	
	To prove (ii), let $v< w_{2k+2}(t,s,y)$ and suppose that there exists $\gamma\in \Gamma_k$, $\assetsProcess\in \K\cap [-k, k]$ and a stopping time $\tau \geq t$ such that 
	$V_{\tau}^{\text{liq}, t, z, \gamma} > w_k(\tau, \S(S^{t,z,\gamma}_\tau, Y^{t,z,\gamma}_\tau,  \assetsProcess^{t,z,\gamma}_\tau), Y^{t,z,\gamma}_\tau-\assetsProcess^{t,z,\gamma}_\tau)$ for $z = (\calS(s,y, -\theta), y+\theta, \theta, v)$. 
	Then by \eqref{eq:ineq bar w and w in DPP 2}, $V_{\tau}^{\text{liq}, t, z, \gamma} > \bar{w}_{k+1}(S^{t,z,\gamma}_\tau, Y^{t,z,\gamma}_\tau,  \assetsProcess^{t,z,\gamma}_\tau)$ and thus, by \cite[proof~of~Thm.3.1, Step~2]{SonerTouzi02}, we get that $v \geq \bar{w}_{2k+1}(t, \calS(s,y, -\theta), y+\theta, \theta)$. In particular, by \eqref{eq:ineq bar w and w in DPP 1} we conclude that $v\geq w_{2k+2}(t, s, y)$, hence a contradiction.
\end{proof}

\begin{remark}
	Part (\ref{part2thm}) of the theorem is stated in terms of $w_k$ instead of $w$ 
	because of a measurable-selection argument 
	employed in the proof, cf.\  \cite[Remark~3.2]{BLZ16a}.
\end{remark}

To derive the pricing PDE from the dynamic programming principle in \cref{thm:GDPP}, we need the dynamics of the continuous processes
\begin{equation}\label{eq:dpp process}
	t \mapsto  V^{\text{liq}}_t - \varphi(t, \calS(S_t, Y^\assetsProcess_t, \assetsProcess_t), \calY(Y^\assetsProcess_t, \assetsProcess_t))
\end{equation}
for sufficiently smooth functions $\varphi: [0,T]\times\RR_+ \times \RR$, $(t,s,y)\mapsto \varphi(t,s,y)$, that will later serve as test functions when characterizing value functions by viscosity solutions.

\begin{lemma}\label{lemma:dynamics of DPP}
	For every $\gamma = (a,b,\nu)\in \Gamma$ and every $\varphi\in C^{1,2,1}([0,T]\times\RR_+ \times \RR)$, we have for $\assetsProcess = \assetsProcess^\gamma$
	\begin{align*}
		\diff ( V^{\text{liq}}_t &- \varphi(t, \calS_t, \calY_t)) = \\
		&\calS_t \left( \frac{F(\calY_t +\assetsProcess_t) - F(\calY_t)}{f(\calY_t)} -\varphi_s \right) \left\{\left((\mu_t - \lambda(\calY_t) h(\calY_t+\assetsProcess_t)\right)\diff t + \sigma \diff W_t\right\}
		\\&+ \left\{ -\varphi_t  - 1/2\sigma^2 \calS_t^2 \varphi_{ss}  +  h(\calY_t+\assetsProcess_t)\varphi_y + \mathfrak{F}(\calS_t, \calY_t, \assetsProcess_t)\right\}\diff t,
	\end{align*}
	with 
	\[
	\mathfrak{F}(s, y, \theta) = s \,h(y+\theta)\left( \lambda(y) \frac{F(y+\theta) - F(y)}{f(y)} - \frac{f(y+\theta) - f(y)}{f(y)}\right),
	\]
	where $\calS_t = \calS(S_t, Y^\assetsProcess_t, \assetsProcess_t)$, $\calY_t = \calY(Y^\assetsProcess_t, \assetsProcess_t)$ and the derivatives of $\varphi$ are evaluated at $(t, \calS_t, \calY_t)$.
\end{lemma}
\begin{proof}
	Since $\calS_t= \calS(S_t, Y^\assetsProcess_t, \assetsProcess_t)$ equals $ \baseS_t f(Y^\assetsProcess_t - \assetsProcess_t)$, the product rule and $f' = \lambda f$ imply 
	\begin{align}
		\diff \calS_t 
		&=  \calS_t \big\{ (\mu_t - \lambda(Y^\assetsProcess_t - \assetsProcess_t) h(Y^\assetsProcess_t))\diff t +\sigma \diff W_t\big\} .  \label{eq:dyn of calS}
	\end{align}
	By It\^{o}'s formula, we obtain
	\begin{align}
		\diff \varphi(t, &\calS_t, Y^\assetsProcess_t - \assetsProcess_t) = \varphi_t \diff t + \varphi_s \diff \calS_t + \varphi_y \diff(Y^\assetsProcess_t - \assetsProcess_t) + 1/2 \varphi_{ss}\diff[\calS]_t \nonumber\\
		&= \left\{ \varphi_t  -\lambda(Y^\assetsProcess_t - \assetsProcess_t) h(Y^\assetsProcess_t) \calS_t\varphi_s - h(Y^\assetsProcess_t)\varphi_y + 1/2\sigma^2 \calS_t^2 \varphi_{ss}\right\} \diff t \nonumber \\
		&\qquad \qquad +  \mu_t \calS_t\varphi_s \diff t+ \sigma \calS_t\varphi_s \diff W_t. \label{eq:key lemma eq 1}
	\end{align}
	With reference to \eqref{eq:inst liq value}, we have 
	\begin{align}
		\diff &V^{\text{liq}}_t = -h(Y^\assetsProcess_t)\calS_t \frac{f(Y^\assetsProcess_t) - f(Y^\assetsProcess_t - \assetsProcess_t)}{f(Y^\assetsProcess_t - \assetsProcess_t)} \diff t \nonumber
		\\&+ \mu_t \calS_t \frac{F(Y^\assetsProcess_t) - F(Y^\assetsProcess_t - \assetsProcess_t)}{f(Y^\assetsProcess_t - \assetsProcess_t)} \diff t+\sigma \calS_t \frac{F(Y^\assetsProcess_t) - F(Y^\assetsProcess_t - \assetsProcess_t)}{f(Y^\assetsProcess_t - \assetsProcess_t)} \diff W_t . \label{eq:key lemma eq 2}
	\end{align}
	Combining \eqref{eq:key lemma eq 1} and \eqref{eq:key lemma eq 2} and rearranging the terms completes the proof.
\end{proof}

\begin{remark}
	Consider the case when $\lambda$ is constant, i.e.~$f = \exp(\lambda \cdot)$. Then we simply have $\mathfrak{F}\equiv 0$
	and the dynamics of $V^{\text{liq}}$ can be stated in a surprisingly simplified form, namely
	\[
	\diff V^{\text{liq}}_t = F(\assetsProcess_t)\diff \calS_t,
	\] 
	where $\calS_t = \calS(S_t, Y^\assetsProcess_t, \assetsProcess_t)$ has the dynamics \eqref{eq:dyn of calS}. As a consequence, the 
	superhedging price (for the large investor) of an option with maturity $T$ and \emph{pure cash settlement} $H(S_T)$ is at least the small investor's price of $H$,  in absence of the large trader, when the price process is $\baseS$ instead.
	Indeed, for each (bounded) superhedging strategy  $\assetsProcess$   (by the large investor) with initial capital $v$ there exists $\PP^\assetsProcess \sim \PP$ (on $\F_T$) such that $\calS = S_{0-}\mathcal{E}(\sigma \widetilde{W})$ under $\PP^{\assetsProcess}$ for a $\PP^\assetsProcess$-Brownian motion  $\widetilde{W}$. Hence, $V^{\text{liq}}(\assetsProcess)$ is a $\PP^\assetsProcess$-martingale and thus $v \geq \EE^{\PP^{\assetsProcess}}[H(S_T)] = \EE^{\PP^{\assetsProcess}}[H(\calS_T)]$ (recall that $\assetsProcess_T = 0$, implying $S_T = \calS_T$). On the other hand, a Feynman-Kac argument shows that $\EE^{\PP^{\assetsProcess}}[H(\calS_T)]$ is just the classical Black-Scholes price for a small investor in a frictionless market with risky asset process $\baseS$. As $\assetsProcess$ was an arbitrary superhedging strategy with initial capital $v$, taking the infimum yields the claim.

	The above observation shows a notable difference to the model in \cite[Thm.~5.3]{BankBaum04}, where the price for the large investor would be typically smaller. This is mainly due to a different specification of superhedging strategies with less stringent settlement constraints, according to which a large trader 
	may be able to reduce at maturity the payoff of the option to a larger extend, by exploiting her price impact on the underlying at maturity. That means,~she can vary at maturity her risky asset position 
	in order to minimize the payoff with less constraints, and immediately afterwards could unwind 
	any residual risky asset position
	at no additional cost (by absence of bid-ask spread). In contrast, our setup is more restrictive by imposing as settlement constraint on the strategies that they have to replicate the physical delivery part exactly, i.e. after settlement the hedging strategy
	has to hold a non-negative cash position without residual holdings in the risky asset.
	
	We note that an argument as above does  not apply  in the general case with non-constant $\lambda$ for our price impact model.  In fact, examples 
	in 
	\cref{sec:Numerics} 
	also reveal situations where superhedging 
	could be 
	cheaper for the large trader, 
	cf.\ \cref{example:superhedging BS is bigger}.
\end{remark}

\section{The pricing PDEs and main results}
\label{sec:The pricing PDEs}

Next, we determine the terminal value for the function $w$ at maturity date $T$, that will serve as a boundary condition for the pricing PDE. 
Recall that $\K$ is the (constraint) set in which trading strategies take values and set $\K_n = \K \cap [-n,n]$ for $n\in\mathbb{N}$.
\begin{lemma}[Boundary condition]\label{lemma:BC}
	For $n\in \mathbb{N}$, let
	\[
	H_n (s,y):= \inf \big\{g_0\big(s\tfrac{f(y+\theta)}{f(y)}, y+\theta \big) + s \tfrac{F(y+\theta) - F(y)}{f(y)} \mid \theta\in \K_n,\ \theta = g_1\big(s\tfrac{f(y+\theta)}{f(y)}, y+\theta\big)\big\}.
	\]
	Then we have $w_n(T,\cdot) = H_n(\cdot)$ and  $w(T,\cdot) = H(\cdot)$, where the function $H$ is given by
	\begin{equation}\label{eq:def of H}
		H := \inf_{n\geq 0} H_n. \tag{\text{\bf{BC}}}
	\end{equation}
\end{lemma}
\begin{proof}
	At maturity time $T$, the hedger of the option has to do a block trade of size $\theta$ in order to meet the physical delivery part specified by $g_1$, thereby moving the price of the underlying from $s$ to $s\tfrac{f(y+\theta)}{f(y)}$ and the impact level from $y$ to $y+\theta$. 
	Such a block trade incurs costs of size $s \tfrac{F(y+\theta) - F(y)}{f(y)}$ and hence it superreplicates the payoff $(g_0, g_1)$ if the hedger can cover this costs and the required  cash delivery part, which after the block trade is $g_0\big(s\tfrac{f(y+\theta)}{f(y)}, y+\theta \big)$.
\end{proof}
\begin{remark}\label{rmk:boundary condition}
	Note that $H(s,y) = +\infty$ holds  if the equation $\theta = g_1\big(s\tfrac{f(y+\theta)}{f(y)}, y+\theta\big)$ does not have a solution  $\theta $ in $ \K$.
\end{remark}

As  we do not  know at this point whether the value function $w$ is continuous, we need to work with discontinuous viscosity solutions  and hence to consider the relaxed semi-limits 
\begin{equation}\label{eq:def of lower semilimit}
	w_*(t, s, y):= \liminf_{(t',s',y',k)\rightarrow(t,s,y,\infty)}w_k(t',s',y'),
\end{equation}
\begin{equation}\label{eq:def of upper semilimit}
	w^*(t, s, y):= \limsup_{(t',s',y',k)\rightarrow(t,s,y,\infty)}w_k(t',s',y'),
\end{equation}
where the limits are taken over $t' < T$.
Recall that $w$ is a (discontinuous) viscosity solution (of our pricing equations, see \cref{subsec:bounded impact function,subsec:The case of constant lambda}) if $w_*$ (resp.~$w^*$) is a supersolution (resp.~subsolution). 
For proving the viscosity property
we make the following assumption.
%
%
\begin{assumption}\label{as:semicontinuous limits are finite}
	
	\begin{description}
		\item[(Boundness):] The functions  $w_*$ and $w^*$ are bounded on $[0,T]\times \RR_+\times \RR$; 
		\item[(Regular payoff):] $H$ from \eqref{eq:def of H} is continuous, bounded, and $H_n \downarrow H$ uniformly on compacts.
	\end{description}
	
\end{assumption}

In particular, \cref{as:semicontinuous limits are finite} implies that $w(T,\cdot)$ is finite. This means that the payoff is well-behaved in terms of the physical delivery part, i.e.~if the trader was supposed to fulfill his obligation from selling the option immediately, he would be able to do so in any situation (in any state of $(s,y)$) with an admissible trade, provided that he has enough capital.

\subsection{Case study for a general bounded price impact function $f$}
\label{subsec:bounded impact function}

In this section, the following assumption is supposed to hold.
\begin{assumption}\label{assumption:bounded f}
	The resilience function $h$ is Lipschitz and bounded, the price impact function $f$ is bounded away from 0 and $\infty$, i.e.~$\inf_\RR f > 0$ and $\sup_{\RR} f < +\infty$, $\lambda$ is bounded and continuously differentiable with bounded derivative, and $\K = \RR$ (no delta constraints).
\end{assumption}
Under \cref{assumption:bounded f}, the antiderivative $F$ from \eqref{eq:def of F} and its inverse $F^{-1}$ are   bijections $\RR\to \RR$ and 
Lipschitz continuous with Lipschitz constants $\sup_{\RR} f < +\infty$ and $1/\inf_\RR f$, respectively. 

To derive the pricing PDE just formally (at first, to be justified later) in this case, let $(t, s, y)\in [0,T)\times \RR_+\times \RR$ and apply formally part (i) of DPP in \cref{thm:GDPP} to $v = w(t,s,y)$ (assuming that the infimum in the definition of $w$ is attained) and $\tau = t+$, together with \cref{lemma:dynamics of DPP} for $\varphi = w$, assuming that $w$ is smooth enough. Thus we get the existence of $\theta^{*}$ such that
\begin{align*}
	0\leq &s \big(\tfrac{F(y+\theta^{*}) - F(y)}{f(y)} - w_s(t,s,y)\big) \big\{(\mu_t - \lambda(y)h(y+\theta^{*})) \diff t + \sigma \diff W_t\big\} \\
	&+ \big\{-w_t(t,s,y) - \tfrac{1}{2} \sigma^2 s^2 w_{ss}(t,s,y) + h(y+\theta^{*})w_y(t,s,y) + \mathfrak{F}(s,y,\theta^{*})\big\} \diff t.
\end{align*}
Still arguing just at a formal level, this cannot hold unless 
\begin{align}
	&F(y+\theta^{*}) = f(y)w_s(t,s,y) + F(y) \qquad \text{and}\nonumber\\
	-w_t(t,s,y) - &1/2 \sigma^2 s^2 w_{ss}(t,s,y) + h(y+\theta^{*})w_y(t,s,y) + \mathfrak{F}(s,y,\theta^{*})\geq 0.\label{eq:formal pde gen f}
\end{align}
In particular, $\theta^{*} = \theta^{*}(t,y,s) = F^{-1}\big( f(y)w_s(t,s,y) + F(y)\big) - y$. The second part of DPP in \cref{thm:GDPP} will actually give that the drift term must be 0, i.e.~we should have equality in \eqref{eq:formal pde gen f}. This formally motivates that the form of the pricing PDE for $w$ should be
\begin{equation}\label{eq:pricing pde bounded f}
	-w_t -\frac{1}{2}\sigma^2 s^2w_{ss} + \tilde{h}(t,s,y) (w_y +s\lambda(y)w_s) + s \tilde{h}(t,s,y) (1 - \tfrac{\tilde{f}(t,s,y)}{f(y)}) = 0, \tag{\text{\bf{PDE}}}
\end{equation}
where for $(t,s,y)\in [0,T)\times \RR_+\times \RR$
\begin{align*}
	\tilde{h}(t,s,y)&:= h\circ F^{-1}\big(f(y)w_s(t,s,y) + F(y)\big), \\
	\tilde{f}(t,s,y)&:= f\circ F^{-1}\big(f(y)w_s(t,s,y) + F(y)\big).
\end{align*}
Observe that the PDE is semilinear and degenerate (since not containing  second order derivatives involving the  $y$-variable).
Our main result is as follows.
\begin{theorem}\label{thm:pde for general f}
	Under \Cref{as:semicontinuous limits are finite} and \Cref{assumption:bounded f}, the
	value function 
	$w$ 
	of the superhedging problem
	is continuous and 
	is 
	the unique bounded viscosity solution to \eqref{eq:pricing pde bounded f} with the boundary condition $w(T,\cdot) = H(\cdot)$, where $H$ is defined in \eqref{eq:def of H}.
\end{theorem}
\begin{proof}
	The viscosity property, i.e.~that $w_*$ (respectively~$w^*$) is a viscosity supersolution (respectively~subsolution), follows by the dynamic programming principle in \cref{thm:GDPP} together with \cref{lemma:dynamics of DPP}. The key arguments are presented in \Cref{subsec:Proof of viscosity solution properties} in detail for the case where $\lambda$ is constant, which actually leads to a slightly more involved pricing PDE \eqref{eq:pde const lambda} (including gradient constraints) requiring additional justifications.
	
	The comparison result of \cref{thm:comparison multipl general} proves uniqueness and continuity, cf.~\cref{rmk:proof of uniqueness}.
\end{proof}

Let us conclude this section by commenting on some consequences from \cref{thm:pde for general f} for the
superhedging price and the existence of a 
respective  hedging strategy. A numerical example is presented in \cref{sec:Numerics}.

\begin{remark}
	Like in the classical case of liquid markets (without price impact), the 
	superhedging price does not depend on the drift in the unperturbed price process. This may be seen more directly by working under the equivalent martingale measure for $\baseS$ from the beginning. On the other hand, the 
	superhedging price   depends non-trivially on the initial level of impact $y$ and the resilience function $h$, and can do so
	even for option payoffs of the form $(g_0(s), 0)$, i.e.~not depending on the level of impact. So it turns out that for the pricing and hedging (cf.~\cref{rmk:replicating strategy general f}) the deviation of the market price from the `unaffected' value, determined by the impact level $y$, is a relevant state variable.
	
\end{remark}

\begin{remark}
	Observe that for only permanent impact, that means~for $h\equiv 0$, \eqref{eq:pricing pde bounded f} simplifies to the  classical (frictionless) Black-Scholes pricing equation and hence the
	superhedging price for the large trader equals the  Black-Scholes price for the option with payoff $H$.
\end{remark}

\begin{remark}
	\label{rmk:replicating strategy general f}
	Under sufficient regularity, it turns out  that a strategy can be constructed that is perfectly replicating the option payout from 
	the (minimal) superhedging price. This means, we have dynamic hedging in the sense of replication, like in the frictionless complete   Black-Scholes model.
	
	To this end, suppose that a function $w\in C^{1,3,1}_b([0,T]\times \RR_+ \times \RR)$ solves the pricing PDE \eqref{eq:pricing pde bounded f} with the boundary condition $w(T,\cdot) = H(\cdot)$.
	Then for any $\varepsilon>0$ a superhedging strategy with an initial cost of $w(0,s,y)+\varepsilon$ can be constructed as follows. Consider the self-financing strategy $(\assetbeta, \assetsProcess)$ with $\assetbeta_{0-} = w(0,s,y)+\varepsilon$, $\assetsProcess_0  = F^{-1}(f(y)w_s(0,s,y) + F(y)) - y$, meaning that a block trade of size $\Delta\assetsProcess_0 = \assetsProcess_0$ is performed at time 0, and
	\begin{align}
		&\assetsProcess_t =  F^{-1}\left(f(\calY^\assetsProcess_t)w_s(t,\calS(S_t, Y^\assetsProcess_t, \assetsProcess_t),\calY^\assetsProcess_t) + F(\calY^\assetsProcess_t)\right) - \calY^\assetsProcess_t \quad \text{for }t\in [0,T),\label{eq:replicating strategy gen case}\\
		&\assetsProcess_T = 0, \quad \text{i.e.}\quad \Delta \assetsProcess_T = \assetsProcess_{T-}, \label{eq:replicating strategy gen case BC}
	\end{align}
	where $\calY^\assetsProcess = Y^\assetsProcess - \assetsProcess$. Then by \cref{lemma:dynamics of DPP}, together with \eqref{eq:replicating strategy gen case} and \eqref{eq:pricing pde bounded f} we conclude that  
	\begin{align*}
		\varepsilon&= V^{\text{liq}}_0(\assetsProcess) - w(0, s, y) = V^{\text{liq}}_T(\assetsProcess) - w(T, \calS(S_T, Y^\assetsProcess_T, \assetsProcess_T), \calY^\assetsProcess_T) \\
		&= V^{\text{liq}}_T(\assetsProcess) - H(\calS(S_T, Y^\assetsProcess_T, \assetsProcess_T), \calY^\assetsProcess_T)\\
		&= V^{\text{liq}}_T(\assetsProcess) - H(S_T, Y^\assetsProcess_T),\qquad \assetsProcess_T = 0,
	\end{align*}
	where the last line follows from \eqref{eq:replicating strategy gen case BC}.
	By definition of $H$, having $H+\varepsilon$ in cash  at time $T$  will be enough to superreplicate the European claim with payoff $(g_0, g_1)$ by doing a possible additional final block trade of size $\Delta^\varepsilon$. Note that such a block trade would not affect $V^{\text{liq}}_T$. Hence,  the strategy $\assetsProcess + \indicator_{\{T\}}\Delta^\varepsilon$ will be superreplicating for the European claim. Note that one could take $\varepsilon = 0$ if the constructed strategy is bounded and the infimum in the definition of $H_n$ is attained (cf.~\cref{lemma:BC}), i.e.~we have a replicating strategy in this case.
	
	An application of It\^{o}'s formula gives that a strategy $\assetsProcess$ satisfying the fixed-point problem \eqref{eq:replicating strategy gen case}  can be obtained, under suitable regularity, by solving the 
	following system of SDEs
	\begin{equation}\label{eq:repl strategy system}
		\begin{aligned}
			&\diff \calS_t = \calS_t[ (\mu_t - \lambda(\calY^\assetsProcess_t) h(\calY^\assetsProcess_t + \assetsProcess_t)) \diff t + \sigma \diff W_t], \\
			&\diff \assetsProcess_t = a(t, \calS_t, \calY^\assetsProcess_t, \assetsProcess_t) \diff t +  b(t, \calS_t,  \calY^\assetsProcess_t) \diff W_t,\\
			&\diff \calY^\assetsProcess_t = -h(\calY^\assetsProcess_t + \assetsProcess_t) \diff t,
		\end{aligned}
	\end{equation}
	with initial conditions $\calS_0 = s$, $\calY^\assetsProcess_0 = y$ and $\assetsProcess_0 = F^{-1}(f(y)w_s(0,s,y) + F(y)) - y$, where
	\begin{align*}
		&a(t, s, y, \theta) := h(y+\theta)\left(1- \frac{\lambda f w_s - f - w_{sy} - \lambda s w_{ss}}{f(F^{-1}(f w_s + F))}\right) + \frac{w_{ts} + s\mu_t w_{ss} + 1/2\sigma^2 s^2 w_{sss}}{f(F^{-1}(f w_s + F))},\\
		&b(t,s,y):= \frac{\sigma s w_{ss}}{f(F^{-1}(f w_s + F))},
	\end{align*}
	and where we write $f = f(y), \lambda = \lambda(y) $, etc., when arguments of functions have not been specified, to  ease the notation. Thus, an optimal superhedging strategy accounts for the transient nature of price impact.
	%
	%
\end{remark}

\begin{remark} To describe how replicating hedging strategies in our model are decribed by coupled Forward-Backward SDEs,
	suppose that $\assetsProcess$ is a replicating strategy for an option with cash-equivalent payoff $H$ and let $(\calY, \calS)$ be the effective impact and price processes. By a change of measure argument, we can assume w.l.o.g.~that $\mu = 0$.
	Setting $Z_t:= \sigma \calS_t \tfrac{F(\calY_t+\assetsProcess_t) - F(\calY_t)}{f(\calY_t)}$, giving $\assetsProcess_t = F^{-1}\left( \sigma^{-1} \calS_t^{-1} f(\calY_t)Z_t+F(\calY_t)\right) - \calY_t$, and using \eqref{eq:key lemma eq 2} leads to the following coupled FBSDE:
	\begin{align*}
		&\diff \calY_t = -(h\circ  F^{-1})\left( \sigma^{-1} \calS_t^{-1} f(\calY_t)Z_t+F(\calY_t)\right)\diff t,  \\
		&\diff \calS_t = \calS_t[ - \lambda(\calY_t) (h\circ  F^{-1})\left( \sigma^{-1} \calS_t^{-1} f(\calY_t)Z_t+F(\calY_t)\right) \diff t + \sigma \diff W_t],\\
		&\diff V^{\text{liq}}_t =\mathfrak{g}(\calY_t, \calS_t, Z_t) \diff t + Z_t \diff W_t, \quad V^{\text{liq}}_T = H(\calS_T, \calY_T),
	\end{align*}
	where the driver of the FBSDE  $\mathfrak{g}:\RR\times \RR_+\times \RR\to \RR$ is given by 
	\[
	\mathfrak{g}(y, s, z) =  -s (h\circ F^{-1})\left( \sigma^{-1} s^{-1} f(y)z+F(y)\right)  \frac{(f\circ F^{-1})\left( \sigma^{-1} s^{-1} f(y)z+F(y)\right)  - f(y)}{f(y)}.
	\]
	
\end{remark}

\begin{example}
	As instructive example, consider an option with maturity $T>0$ whose payout at maturity is the spot price of the asset, i.e.~$H(s,y) = s$. In the frictionless Black-Scholes model its arbitrage-free price is $v^{\text{BS}}(s) = s$ and a (minimal) replicating strategy is to buy one share at initiation and hold it until maturity, where it is liquidated at the spot price.
	For the solution in our price impact model, let us consider the classical solution to \eqref{eq:pricing pde bounded f} with the boundary condition $H$ being given by the function 
	\begin{equation}\label{eq:val fn buy and hold}
		w(t,s,y) = \tfrac{F(y+c(t,y)) - F(y)}{f(y)}s,
	\end{equation} 
	where $c:[0,T]\times \RR\to \RR$ is a solution to the following backwards transport equation 
	\[
	\begin{cases}
		-c_t +h(y+c)c_y &= 0 \qquad \text{on }[0,T)\times \RR,\\
		c(T,y) &= F^{-1}(f(y) + F(y))-y \qquad \text{on }\RR.
	\end{cases}
	\]
	In particular, by the dynamics of $c$ it then holds for any strategy $\assetsProcess$ than $c(t,\mathcal{Y}^\assetsProcess_t) = c(0,\mathcal{Y}^\assetsProcess_0)$ for $t\in [0, T]$, where $\mathcal{Y}^\assetsProcess$ is the effective impact process corresponding to $\assetsProcess$.
	In particular, by \eqref{eq:replicating strategy gen case} a minimal replicating strategy satisfies on $[0,T)$ the equation
	\[
	\assetsProcess^*_t = c(t, \mathcal{Y}^{\assetsProcess^*}_t) = c(0, \mathcal{Y}^{\assetsProcess^*}_0)  =  c(0,Y_{0-}).
	\]
	Hence, a buy-and-hold strategy is also optimal for the large trader. We can observe:

	1.) Purely permanent impact (that means $h = 0$) would yield the Black-Scholes price $w(t,s,y) = s$ and the buy-and-hold strategy of $c(0,y) = c(T,y) = F^{-1}(f(y) + F(y))-y$ shares, that does not depend on the maturity $T$.

	2.) In comparison, if price impact is not permanent 
	but transient ($h\neq 0$), the price \eqref{eq:val fn buy and hold} depends non-trivially on the maturity $T$, in addition to the price impact and resilience functions $f$ and $h$ respectively.

	3.) The large trader's price  $w(t,s,y)$ dominates the Black-Scholes price $v^{\text{BS}}(t,s) = s$, if and only if $c(t, y) > c(T, y)$. Moreover, there are situations where this condition holds and situations where it is violated. The reason is that there are two counterbalancing effects: At initiation where the large trader buys shares to set up the initial delta hedge,  moving prices in an unfavourable direction, and at maturity when she liquidates the delta and could move prices in a direction favourable to her. Which of these two effects dominates depends non-trivially on the level of liquidity initially and at maturity, and the settlement specificaions of the option; See discussion in \cref{example:superhedging BS is bigger}.
	
\end{example}

Let us comment here on \cref{assumption:bounded f} that implies bijectivity of $F$ on $\RR$. Observe that its inverse $F^{-1}$
is used to describe the optimal control $\theta^*$. Similar conditions are also crucial for the results in \cite{BankBaum04}  and  \cite{BLZ16a}: See the surjectivity assumption	(A5) in \cite{BankBaum04} and the invertibility assumption (H2) in \cite{BLZ16a}. The next section shows how departing from this assumption leads naturally to singularity in the pricing PDE with respect to the gradient. Indeed, the lack of invertibility of $F$ requires  conditions on $w_s$ so that $\theta^*$ could be derived. Therefore, the analysis there will involve constraints on the `delta', that means on the holdings in the risky asset, what in PDE terms translates to constraints on the spacial gradient $w_s$.

\subsection{Case study for price impact of exponential form}
\label{subsec:The case of constant lambda}

We extend the analysis to a natural case where the antiderivative of the price impact function is 
not assumed to be surjective. To this end, the price impact function is taken to be of exponential form $f(x) = \exp(\lambda x)$ 
with $\lambda$ being a constant (i.e.\  $\log f$ is linear),
meaning that the relative marginal price impact function  $\lambda=f'/f > 0$ is constant. 
A distinctive feature of this case is that at any time 
$t$, knowing the (marginal) stock price $S_t$ 
is sufficient to determine the impact from an instant block trade, since after a block trade of size $\Delta$ the price would be $\baseS_t f(Y_t+\Delta) = S_t \exp(\lambda \Delta)$. Hence, the relative displacement $f(Y^\assetsProcess)$ of $S$ from the fundamental price $\baseS$ is immaterial to determine the 
price impact from a block trade, in difference to the situation of \cref{subsec:bounded impact function}. Motivated by \cref{rmk:permanent impact and constant lambda}, we impose
short-selling constraints, by requiring
trading strategies to evolve in $\K = [-K,\infty)$ for some $K > 0$. 

To derive (only heuristically at first, we will justify it rigorously later) the pricing PDE, 
let us apply formally \cref{thm:GDPP} for $v = w(t,s,y)$ at $t,s,y,$ $\tau = t+$, provided that $w$ is smooth enough,  to get the existence of $\theta^*\in \K$ such that, using \cref{lemma:dynamics of DPP}, we have
\begin{equation}\label{eq:const lambda heuristic pde}
	\mathcal{L}^{\theta^*}w(t, s, y) \diff t - s ( w_s(t,s,y) -e^{\lambda \theta^*}/\lambda + 1/\lambda) (\sigma\diff W_t + \eta_t \diff t) \geq 0,
\end{equation}
where $\eta_t = \mu_t - \lambda h(y+\theta^*)$ and 
\[
\mathcal{L}^{\theta^*}w(t,s,y) := -w_t(t,s,y) + h(y+\theta^*)w_y(t,s,y) - \frac{1}{2}\sigma^2 s^2 w_{ss}(t,s,y).
\]
As in \cref{subsec:bounded impact function}, the diffusion part in \eqref{eq:const lambda heuristic pde} should vanish, giving the optimal control
\[
\theta^* = \frac{1}{\lambda} \log \big(\lambda w_{s}(t,s,y) + 1\big),
\]
and from the drift part we identify the pricing PDE 
$ \mathcal{L}^{\theta^*}w(t,s,y) =0.$  The constraint $\theta^*\in \K$ is now  equivalent to $\mathcal{H}_\K w(t,s,y) \geq 0$, where for a smooth function $\varphi$ we set
\[
\mathcal{H}_\K \varphi(t,s,y) := \lambda \varphi_s(t,s,y) + 1 - e^{-\lambda K}.
\] 
Thus we conclude, just  formally, that $w$ should be a solution to the variational inequality
\begin{equation}\label{eq:pde const lambda}
	\mathcal{F_\K}[w] := \min \{\mathcal{L}^{\theta[w]}w\ ,\ \mathcal{H}_{\K}w \} = 0 \quad \text{on }[0,T)\times \RR_+ \times \RR,  \tag{$\textbf{PDE}^\delta$}
\end{equation}
where 
\begin{equation}\label{eq:opt strategy const lambda}
	\theta[w](t,s,y) :=1/\lambda\cdot  \log \big(\lambda w_s(t,s,y) +1\big).
\end{equation}
As usual, the gradient constraints propagate to the boundary, meaning that the 
boundary condition for \eqref{eq:pde const lambda} 
should be 
\begin{equation}\label{eq:pde const lambda BC}
	\min\{w(T,\cdot) - H, \mathcal{H}_{\K} w\} = 0. \tag{$\textbf{BC}^\delta$}
\end{equation}
After this motivation, we state the main result for exponential price  impact  
$f = \exp(\lambda \,\cdot)$.
\begin{theorem}\label{thm:pricing pde constant lambda}
	Suppose that the resilience function $h$ is Lipschitz continuous and  \cref{as:semicontinuous limits are finite} holds. 
	Then the 
	value function 
	$w$ 
	of the superhedging problem
	is continuous and is the unique bounded viscosity solution to the variational inequality \eqref{eq:pde const lambda} with boundary condition \eqref{eq:pde const lambda BC}. 
\end{theorem}

\begin{proof}
	The technical proofs are deferred to \cref{subsec:Proof of viscosity solution properties}. The viscosity super-/sub-solution property are proven in \cref{thm:supersolution} and \cref{thm:subsolution} respectively, while uniqueness and continuity follow from the comparison result of \cref{thm:comparison}, cf.~\cref{rmk:proof of uniqueness}.
\end{proof}

\begin{corollary}\label{cor:BSpde}
	In the setup from \cref{thm:pricing pde constant lambda}, suppose moreover that the payoff function $(g_0, g_1)$ is not depending on the level of impact $y$ but just on the price $s$  of the underlying. Then the
	superhedging price is a function in $(t,s)$ only and the pricing PDE \eqref{eq:pde const lambda} 
	simplifies to a Black-Scholes PDE with gradient constraints.
	In this case, if the face-lifted payoff
	\[
	F_{\K}[H](s):= \sup_{x \leq 0} \left \{H(s+x) + \frac{1-e^{-\lambda K}}{\lambda} x\right \},\quad s\in \RR_+,
	\]
	is continuously differentiable in $s$ with bounded derivative,
	with the convention that $H = H(0)$ on $(-\infty, 0]$, then the
	superhedging price (for the large trader) coincides with the friction-less Black-Scholes price for the face-lifted payoff $F_{\K}[H]$.
\end{corollary}
\begin{proof}
	If $(g_0, g_1)$ is a function of the price $s$ of the underlying only (but not of $y$), then it is easy to see that $H$ is such as well and that the dimension of the state process can be reduced by 
	omitting the impact process $Y$. In this case, the stochastic target problem in \cref{sec:stoch target problem} could be formulated for the new state process and thus the value function would be a function on $(t,s)$ only. The same analysis can be carried over to derive the pricing PDE and to prove viscosity solution property of the value function. The pricing PDE in this case would be the Black-Scholes PDE with gradient constraints since the term $h(Y)\varphi_y$ in \cref{lemma:dynamics of DPP} would not be present. Hence, the 
	superhedging price in our large investor model would coincide with the 
	superhedging price under delta constraints in the small investor model for the payoff $H$ (because it solves the same PDE).  
	In this one-dimensional setup, this price coincides with the Black-Scholes price for the face-lifted payoff $F_{\K}[H]$, cf.~\cite[Proposition~3.1]{CEK15_delta}.
\end{proof}



\section{Numerical examples}
\label{sec:Numerics}
We discuss 
numerical calculations of the
superhedging price $w$ characterized by \eqref{eq:pricing pde bounded f}, cf.~\cref{thm:pde for general f}, to illustrate 
results. For the computations we consider an impact function
\begin{equation}\label{eq:impact for numerics}
	f(x) = 1 + \arctan(x)/10, \quad x\in \RR,
\end{equation}
satisfing \cref{assumption:bounded f}.
Note that  $\lambda(x) = 1/(10(1+x^2)f(x))$ varies most within the range of about $(-4, 4)$; Here, the change in impact is significant, see \cref{fig:f and lambda}. Apart from satisfying our assumptions and having  $F(x)= x+(x\arctan(x) - 1/2 \log(1+x^2))/10$ in explicit form, being useful for the implementation, it turns out that similar shape of impact
has been observed in the calibration of a related propagator model to real data, see \cite[Appendix]{BussetiLillo12}.

For $h(y) = \beta y$ with $\beta = 1$, we compare the large trader's price of a European call option with physical delivery at maturity $T = 0.5$ and strike $K = 50$, and the option's frictionless price, i.e.~the classical Black-Scholes  price of a European call option for the same model parameters. Let us recall that the case $f=1$ in our price impact model coincides with the Black-Scholes model. The volatility $\sigma$ is set to $0.3$.  The payoff for the large trader is
$H(s,y) = \left(  s\tfrac{F(y+1)-F(y)}{f(y)}-K \right) \indicator_{\{s\geq K\}}$ that we ``smooth out'' by approximating the indicator function from above by linearly interpolating $0$ and $1$ between $K-0.5$ and $K$.

To approximate both prices, we solve the corresponding PDEs using  a (semi-implicit) finite difference scheme in the bounded region $(y,s)\in [-20, 20]\times[0, 200]$. For our simulation we set the following boundary condition for $t<T$: $\tfrac{\partial w}{\partial s} = \left( F(y+1) - F(y)\right)/f(y)$ on $[-20,20]\times\{200\}$, $\tfrac{\partial w}{\partial y} = 0$ on $\{-20, 20\}\times[0,200] \cup [-20,20]\times \{0\}$. Indeed, for initial impact $y$ close to -20 or +20 the impact function is approximately constant and until maturity $T$ resilience would be unlikely to bring back the level of impact to the region where the changes in $f$ are significant, see \cref{fig:f and lambda}; Thus we might expect that the price would not depend that much on the level of impact. On the other hand, for larger values of $s$ one may expect the price to depend approximately linearly on $s$ (like the payoff profile). 
The difference between the Black-Scholes price and the large trader's price (as a function of the risky asset price $s$ and the level of impact $y$) is shown in \cref{fig:call physical}. Let us point out that the Black-Scholes price does not depend on level of impact $y$.

\begin{figure}[!ht]%
	\centering
	\begin{subfigure}{0.45\textwidth}
		\centering
		\includegraphics[width=\textwidth]{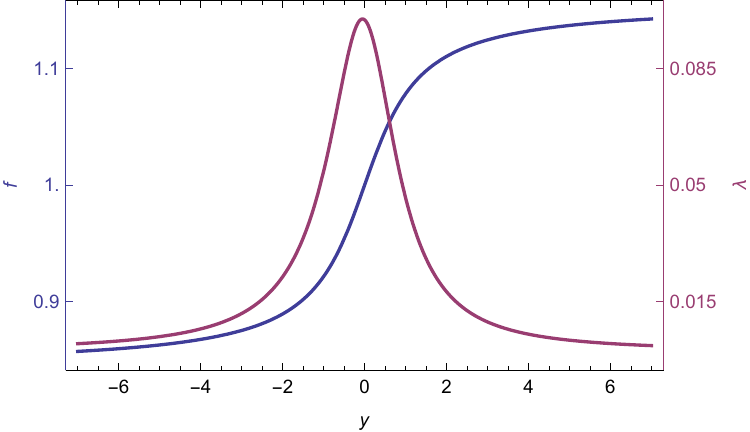}
		\caption{Impact function $f$ (in blue) and its logarithmic derivative $\lambda$ (in purple)}
		\label{fig:f and lambda}
	\end{subfigure}
	~
	\begin{subfigure}{0.45\textwidth}
		\centering
		\includegraphics[width=\textwidth]{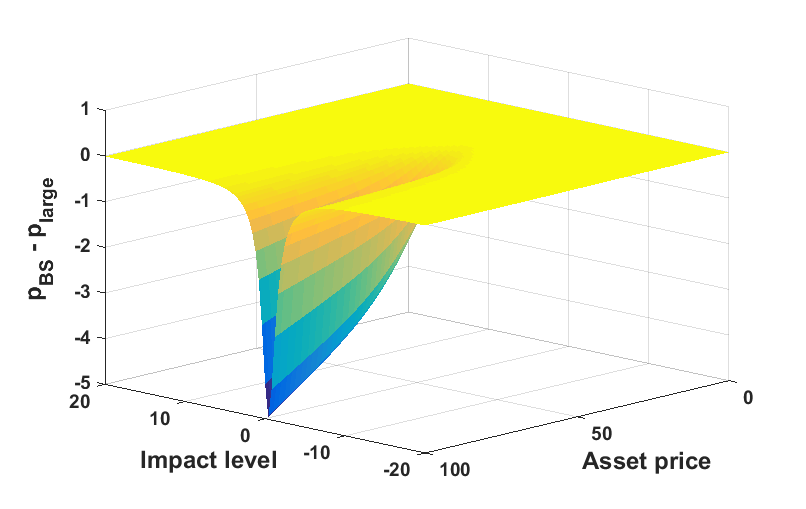}
		\caption{The difference $p_{\text{BS}} - p_{\text{large}}$  of prices of a European call option with physical delivery,  resilience rate is  $\beta = 1$}
		\label{fig:call physical}
	\end{subfigure}
	~	
	\begin{subfigure}{0.45\textwidth}
		\centering
		\includegraphics[width=6.4cm,height=5.4cm]{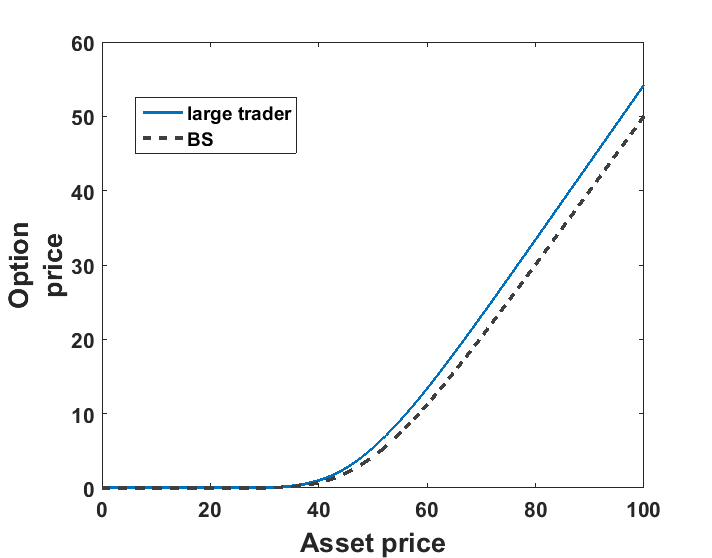}
		\caption{The frictionless Black-Scholes price and the large trader's price for call option with physical delivery, resilience rate $\beta = 1$ and initial impact level $y=0$}
		\label{fig:call physical vs call BS for 0 impact}
	\end{subfigure}
	~
	\begin{subfigure}{0.45\textwidth}
		\centering
		\includegraphics[width=6cm,height=5cm]{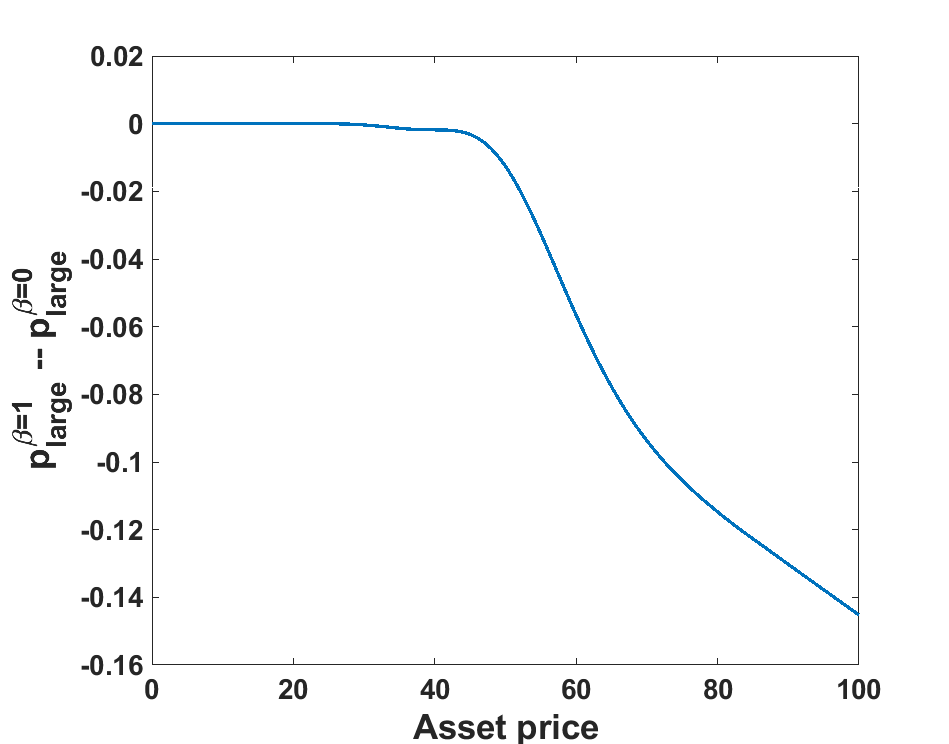}
		\caption{Difference between large trader's prices $p_{\text{large}}^{\beta = 1}$ and $p_{\text{large}}^{\beta = 0}$ for a call option with physical delivery with ($\beta = 1$) and without  ($\beta = 0$) resilience, for initial impact level $y=0$}
		\label{fig:call physical with res vs without res}
	\end{subfigure}
	\caption{Numerical computations with impact function $f$ from \eqref{eq:impact for numerics},  $\sigma = 0.3$, $T=0.5$, strike $K=50$, resilience function $h(y) = \beta y$}
	\label{fig:numerics}
\end{figure}

Numerical results of \cref{fig:call physical vs call BS for 0 impact} illustrate that the superreplication price for the large trader dominates the frictionless  Black-Scholes price for the call option with physical delivery. But  we note  that such property does not need to hold in general. For instance,
it does not appear to do so in case of a European call with pure cash delivery, where numerical computations 
show that for the large investor the price could also be smaller, typically if the impact level at inception is away from zero, see also \cref{example:superhedging BS is bigger}.
The intuition for this more complex behaviour is that for pure cash delivery, the net turnover until maturity  of traded assets for a (super-)hedging strategy must be zero (as $\assetsProcess_{0-} = \assetsProcess_T = 0$ then), while non-zero resilience ($h\neq 0$) induces additional drift, that also turn out to be less costly to the large trader,
if moving the underlying price paths into regions with 
lower (or zero) 
option payout.

On the other hand, superhedging becomes more expensive for the large trader when she has to deliver the underlying asset physically at maturity, since, if the call option settles in-the-money, she needs to do a final block trade to buy what is lacking (in the pre-terminal delta position) for the one physical unit required. But this last price impact at maturity is costly in that it further increases the issuer's call option payout for physical delivery, in comparison to cash settlement, where selling the long delta position decreases the payout.  

In addition, observe that the presence of resilience renders the level of impact (or the displacement from the fundamental price)  
to be a relevant state variable for the problem. For the setup of our numerical example for instance, the price of  a European call option with physical delivery, when hedging is initiated at neutral impact level ($y=0$), is cheaper in the presence of resilience than in the case of no resilience, i.e.~only permanent impact, see \cref{fig:call physical with res vs without res}. This is however not always the case, for example if impact at initiation is negative ($y<0$).
To conclude,  the dependence in $y$ of the option's price is complex: apart from the drift on the prices that the level of impact induces, it also determines the price impact from intermediate tradings and the final trade (enforced by settlement rules).  And we have mentioned examples where superhedging could be  less or more expensive for the large investor in the presence or absence of resilience.

\begin{example}\label{example:superhedging BS is bigger}
	The price of an European option in the Black-Scholes model (for s small investor) could indeed be greater than the 
	superhedging price for the large trader of this option with pure cash delivery. To see this, consider for maturity $T>0$ the solution  $v^{\text{BS}}$ of the Black-Scholes PDE with bounded and smooth terminal condition $H$ that has bounded derivatives, where we moreover assume that $\partial_S H \geq 0$, for instance a smooth approximation of a bull call spread option. Note that in particular $\partial_S v^{\text{BS}} \geq 0$ and the derivatives of $v^{\text{BS}}$ are bounded. We compare now  $v^{\text{BS}}(0, \cdot)$ with $v(0, \cdot, y)$ for large values of $y$, where $v=w$ with $w$ from \cref{thm:pde for general f} with terminal condition $H$. Note that when $y=Y_{0-} > 0$ the affected price process includes additional drift in favorable direction for the large trader. 
	
	Let $\assetsProcess$ with $\assetsProcess_{0-} = 0$ be such that $\assetsProcess_{T} =0$ (corresponding to only cash delivery at maturity) and for $t\in[0, T-]$ 
	\begin{equation}\label{eq:LT from BS strategy}
		\assetsProcess_t = F^{-1}(\partial_S v^{\text{BS}}(t, \calS_t) f(\calY^\assetsProcess_t) + F(\calY^\assetsProcess_t)) - \calY^\assetsProcess_t,
	\end{equation}
	where $\calY^\assetsProcess = Y^\assetsProcess - \assetsProcess$ and $\calS = f(\calY^\assetsProcess)\baseS$. Since $v^{\text{BS}}$ is smooth, the arguments in \cref{rmk:replicating strategy general f} ensure existence of such $\assetsProcess$, while positivitity of  $\partial_S v^{\text{BS}}$ implies $\assetsProcess \geq 0$ on $[0,T]$. Now for the self-financing portfolio $(\assetbeta, \assetsProcess)$ with initial cash holdings $\assetbeta_{0-} = v^{\text{BS}}(0, S_{0-})$ we have by \eqref{eq:dyn of calS}, \eqref{eq:key lemma eq 2} and \eqref{eq:LT from BS strategy} (recall that $S_{0-} = \calS_0$)
	\begin{align}
		V^{\text{liq}}_T &=  v^{\text{BS}}(0, \calS_0) + \int_{0}^{T}\partial_S v^{\text{BS}}(t, \calS_t) \diff \calS_t \nonumber\\
		&\qquad -\int_{0}^T \calS_t h(Y^\assetsProcess_t ) \left( \frac{f(Y^\assetsProcess_t) - f(Y^\assetsProcess_t - \assetsProcess_t)}{F(Y^\assetsProcess_t) - F(Y^\assetsProcess_t - \assetsProcess_t)} - \lambda(Y^\assetsProcess_t - \assetsProcess_t)\right)\diff t \nonumber \\
		&= H(\calS_T) - \int_{0}^T \calS_t h(Y^\assetsProcess_t ) \left( \frac{f(Y^\assetsProcess_t) - f(Y^\assetsProcess_t - \assetsProcess_t)}{F(Y^\assetsProcess_t) - F(Y^\assetsProcess_t - \assetsProcess_t)} - \lambda(Y^\assetsProcess_t - \assetsProcess_t)\right)\diff t.
		\label{eq:wealth for BS strategy}
	\end{align}
	In particular, if the integrand in \eqref{eq:wealth for BS strategy} is negative on $[0,T]$, then $(\assetbeta, \assetsProcess)$ would be a superhedging strategy for the large trader with initial capital $\assetbeta_{0-} = v^{\text{BS}}(0, S_{0-})$ and hence 
	\begin{equation}\label{eq:BS price is bigger}
		v(0, S_{0-}, Y_{0-}) \le v^{\text{BS}}(0, S_{0-}).
	\end{equation}
	One can show that the integrand will be negative for instance when $Y^\assetsProcess \ge 0$ on $[0,T]$ and $\lambda$ is strictly decreasing (at least on a compact set containing the range of $Y^\assetsProcess$ and $Y^\assetsProcess - \assetsProcess$);
	Such a situation could arise if for example $Y_{0-}$ is large enough. 
	Alternatively, a negative integrand could also occur if for instance $Y^\assetsProcess$ is negative on $[0,T]$, for instance if $Y_{0-}$ is small enough, and $\lambda$ is strictly increasing. Let us mention that equality in  \eqref{eq:BS price is bigger} cannot hold in general, for all values of $S_{0-}, Y_{0-}$, as this would imply that $v$ would not depend on the initial level of impact $Y_{0-}$, that is not the case for general payoff functions $H$, see e.g.~\cref{fig:call physical}.

\end{example}


\section{Extensions: Permanent impact, covered options, multiple illiquid assets with cross-impact}
\label{sec:permanent impact}

This section explains possible extensions and variations of the previous results on hedging under multiplicative transient price impact. We show 
first, how the results generalize to 
\emph{combined} transient and permanent price impact, and explain, how working in suitable effective coordinates further enables extensions to \emph{multiple illiquid assets}. We also comment and give references for the solution to the  different but related hedging problem for \emph{covered} options.

For $\eta \geq 0$, the marginal price of the risky asset (for an extra infinitesimal quantity) 
is 
\begin{equation}\label{eq:extension permanent impact}
	S_t := f(\eta \assetsProcess_t + Y^\assetsProcess_t)\baseS_t,
\end{equation}
in a generalization of equation \eqref{eq:price process}, with $Y^\assetsProcess$ being given by \eqref{eq:impact process}.
Following the arguments in \cite[Section~5.4]{BechererBilarevFrentrup-model-properties}, the
(stable) proceeds from a general semimartingale strategy $\assetsProcess$ are
\begin{equation*}
	\tilde{L}(\assetsProcess) := \frac{1}{1+\eta}\left(\int_0^\cdot F(\eta\assetsProcess_t + Y^{\assetsProcess}_t)\diff \baseS_t - \int_0^\cdot \baseS_t f(\eta\assetsProcess_t + Y^{\assetsProcess}_t)h(Y^\assetsProcess_t)\diff t - \left.\baseS F(\eta\assetsProcess + Y^{\assetsProcess})\right\vert_{0-}^\cdot\right).
\end{equation*} 
In particular, a block trade $\Delta \assetsProcess_t$ yields the proceeds $-\baseS_t \tfrac{1}{1+\eta}\int_0^{(1+\eta)\Delta\assetsProcess_t} f(\eta\assetsProcess_{t-} + Y^{\assetsProcess}_{t-} + x) \diff x$. Thus, following the discussion in \cref{sec:impact model}, the volume effect process (in the spirit of \cite{PredoiuShaikhetShreve11}) in this case is $\eta\assetsProcess + Y^\assetsProcess$ and thereby has a permanent and a transient  component.
The dynamics of the instantaneous liquidation value process $\tilde{V}^{\text{liq}}$ now satisfies 
\begin{equation*}
	(1+\eta) \diff\tilde{V}^{\text{liq}}_t = (F(\eta\assetsProcess_t + Y^{\assetsProcess}_t) - F(Y^\assetsProcess_t - \assetsProcess_t))\diff \baseS_t - h(Y^\assetsProcess_t)(f(\eta\assetsProcess_t + Y^{\assetsProcess}_t) - f(Y^\assetsProcess_t - \assetsProcess_t)) \diff t.
\end{equation*}
It is worth noting that the generalization by an additional permanent impact is not changing the effective price and impact processes $\calS(S,Y^\assetsProcess, \assetsProcess)$ and $\calY(Y^\assetsProcess, \assetsProcess)$ because the permanent component vanishes for asset holdings with zero shares in the risky asset. Therefore, the previous analysis carries over to additional permanent impact, with  adjustments as follows:
\begin{itemize}
	\item The boundary condition in \cref{lemma:BC} needs to be modified by adding the prefactor $1+\eta$ to  $\theta$, when  $\theta$ appears as an argument of a function;
	\item In \cref{lemma:dynamics of DPP}, $F(Y^\assetsProcess)$ should be substituted by $F(\eta \assetsProcess + Y^\assetsProcess)$, all the fractions should be divided by $1+\eta$ and $\mathfrak{F}$ will now become $\mathfrak{F}^\eta$ with
	\[
	\mathfrak{F}^{\eta}(s, y, \theta) := sh(y+\theta)\left( \lambda(y) \tfrac{F(y+ (1+\eta)\theta) - F(y)}{(1+\eta)f(y)} - \tfrac{f(y+(1+\eta)\theta) - f(y)}{(1+\eta)f(y)}\right).
	\]
\end{itemize} 
Let us first discuss the setup of \cref{subsec:bounded impact function} which essentially required $F$ to be invertible. In this case, the pricing PDE will have the same structure as \eqref{eq:pricing pde bounded f} with the following modifications $\tilde{h} = \tilde{h}^{\eta}$ and $\tilde{f} = \tilde{f}^\eta$ replacing the former $\tilde{h}$ and  $\tilde{f}$, namely
\[
\tilde{h}^\eta(t,s,y) = h\left(\tfrac{1}{1+\eta} F^{-1}((1+\eta)f(y)\varphi_s(t,s,y) + F(y)) + \tfrac{\eta}{1+\eta}y\right),
\]
\[
\tilde{f}^\eta(t,s,y) = f\circ F^{-1}((1+\eta)f(y)\varphi_s(t,s,y) + F(y)).
\]
An optimal hedging strategy $\assetsProcess^*$, if it exists, satisfies (as in~\cref{rmk:replicating strategy general f} for $\eta = 0$)
\[
(1+\eta)\assetsProcess^*_t = F^{-1}((1+\eta)f(\calY_t^*)\varphi_s(t,\calS_t^*,\calY_t^*) + F(\calY_t^*)) - \calY_t^*,
\]
where $\calS^* = \calS(S, Y^{\assetsProcess^*}, \assetsProcess^*)$ and $\calY^* = \calY(Y^{\assetsProcess^*}, \assetsProcess^*)$. Hence, the large trader's optimal strategy also reflects the permanent component in addition to the displacement from the fundamental price process tracked by $Y^\assetsProcess$.

In the setup from \cref{subsec:The case of constant lambda}, we again  consider portfolio constraints $\theta\in\K = [-K, +\infty)$ in order to derive the pricing PDE. Thanks to $\mathfrak{F}^{\eta} = 0$, the pricing PDE here simplifies to
\[
\min \big\{ -w_t - \frac{\sigma^2}{2} s^2 w_{ss} + h(y+\theta^*)w_y,\ \lambda(1+\eta) \varphi_s + 1 - e^{-\lambda(1+\eta) K}\big\} = 0,\; 
(t,s,y)\in [0,T)\times\RR_+\times \RR,
\]
where $\theta^* = \frac{1}{\lambda(1+\eta)} \log \big(\lambda(1+\eta) w_{s} + 1\big)$, with boundary condition
\[
\min\{w(T,\cdot) - H,\ \lambda(1+\eta) \varphi_s + 1 - e^{-\lambda(1+\eta) K} \} = 0,
\]
where $H$ is the modified boundary condition from \cref{lemma:BC}, as explained above. In particular, the pricing PDE with permanent impact coincides with the pricing PDE with purely transient 
impact but with suitably modified $\lambda$, that in this case becomes $\lambda (1+\eta)$.

\begin{remark}\label{rem:extension-covered-options}
	We explain how to obtain further results and note key differences in the related hedging problem for so-called covered options, as in \cite{BLZ16b} but under multiplicative transient price impact, where our analysis carries over similarly, by adopting arguments of \cite{BLZ16b} from the case of additive permanent price impact, as shown in   \cite[Sect.8]{BilarevBecherer18}.
	
	In contrast to the problem  studied in the main body of the present paper and in \cite{BLZ16a} for \emph{non-covered} options there is no price impact at inception and at maturity in the hedging problem for \emph{covered} options. Such makes the stochastic target problem very different.
	The financial interpretation is, that the buyer of a covered option provides (at discretion of the hedger) 
	the required initial (delta) hedging position  
	and accepts any mix of cash and stocks (at suitable book value if evaluated at
	current marginal market prices $S$) as an option settlement.
	In this way, the hedger is not exposed to initial and terminal impact when forming and unwinding the hedging position for covered options. We mention that similar assumptions are made in the literature \cite{Frey1998,FreyPolte11,CetinSonerTouzi10}
	where analysis is in terms of book value, instead of liquidation value \cite{BankBaum04,BLZ16a}.
	
	In the previous sections, the superhedging price 
	for (non-covered) options under transient multiplicative price impact is characterized by a degenerate semilinear PDE, whose non-linearity involves the resilience function $h$ and the price impact function $f$. It can involve gradient constraints (that means \emph{delta}-constraints), reducing to the Black-Scholes equation with gradient constraints in the situation of \cref{cor:BSpde}.

	In contrast, for  \emph{covered options}, 
	the corresponding pricing equation turns out to be fully non-linear and singular in the second-order term. This induces \emph{gamma} constraints, whereas for non-covered options singularity arises in the first order derivative and induces delta constraints, see  \cref{subsec:The case of constant lambda}. For covered options, it can be shown \cite[Sect.8]{BilarevBecherer18} that the resilience of price impact is immaterial for the hedging price, irrespectively of 
	a particular form for the resilience function,
	what has been observed likewise in  \cite[Section~4]{BLZ16b} for additive impact.
	We emphasize that this is very different to  \cref{subsec:bounded impact function}, where the resilience function enters the pricing equation in a non-trivial way. 
	It turns out that the current deviation of the asset price from the unaffected price becomes a relevant state variable for describing the solution.
	Moreover, one can show \cite[see Remark~8.2, 2)]{BilarevBecherer18} that
	the superheding price is decreasing in the impact function $\lambda$, in the sense that if $\lambda \geq \tilde{\lambda}$, then the price with respect to ${\lambda}$ dominates the one with respect to  $\tilde{\lambda}$. For a dual formulation for the hedging of covered options we refer to \cite{BouchardTan22}.
\end{remark}

\begin{remark}\label{rem:extension-crossimpact}
	Working in effective coordinates, as explained in \cref{sec:stoch target problem}, further permits to extend results about transient price impact, in additive or multiplicative form, to multiple risky assets with cross-impact from transactions across different assets (being described in \cite[ch.5, cf.\ example 5.1.6]{Bilarev18}).
	To this end, a key idea is that the impact function needs to be the gradient field of a suitable potential in order to avoid a form of instantaneously profitable round-trips (cf.\ \cite[Thm.5.1.4]{Bilarev18}).
	Thereby, results like from previous sections (or \cite{BLZ16a} for permanent impact) can be extended to multiple assets in an 
	additive transient cross-impact model. One obtains a geometric DPP and a viscosity PDE to characterize superhedging prices, which does involve the resilience function ($h$) of transient impact (see \cite[Sect.5.3.2]{Bilarev18}). And under certain conditions one recovers
	as instructive reference case again results as in a multi-dimensional Bachelier model with its natural pricing formula \cite{Bilarev18}[in Rem.5.3.8]), that is not involving the price impact. This extends (to multiple dimensions) the instructive 1-dimensional linear permanent impact example from \cite[Sect.2.4]{BLZ16a}, which also yields the familiar Bachelier pricing formula. Notice that the hedging strategy is affected, though closely related to the usual Bachelier-delta-hedging strategy, by being computed at liquidation magnitudes of the stock price. This is entirely analogous to Black-Scholes quantities occuring under (permanent) multiplicative impact the basic log-linear example of our model (see \cref{example:loglinear} and \cref{rem:BLZimpactcomparison}).        
\end{remark}


\begin{appendix}
	
	\normalsize
	
	\section{Proofs}
	\label{subsec:Proof of viscosity solution properties}
	This section provides the proofs delegated from \cref{sec:stoch target problem}, in particular the proof of  \cref{thm:pricing pde constant lambda}. Recall that in this case $f(x) = \exp(\lambda x)$ for $\lambda > 0$ and thus the effective price simplifies to $\calS(s,y,\theta) = se^{-\lambda \theta} \equiv \calS(s,\theta)$, i.e.~the level of impact is not needed in order to determine the price change of a block trade, given the price before the trade. We consider strategies taking values in $\K = [-K, +\infty)$ for $K>0$. This yields a gradient constraint for the PDE  that is needed because of a singularity in the PDE, for the expression 
	\eqref{eq:opt strategy const lambda} for the form of the optimal strategy to be finitely defined.
	
	First, we verify in \cref{subsec:Verification argument} that if the pricing PDE \eqref{eq:pde const lambda} admits a sufficiently smooth classical solution, then a replicating strategy in feedback form can be constructed. Such a construction will be needed also for the contradiction argument in the proof of the subsolution property in \cref{sec:Visc sol property} where, using smooth test functions, one constructs locally strategies which, roughly speaking, behave like replicating strategies. The viscosity property proofs are collected in \cref{sec:Visc sol property} and in \cref{sec:comparison} we prove comparison results that imply uniqueness of the viscosity solutions of the pricing PDEs and continuity of the value function for the superhedging problem.
	
	\subsection{Verification argument for exponential impact function}
	\label{subsec:Verification argument}
	Suppose that $w\in C^{1,2,1}([0,T]\times \R_+\times \R)$ is such that for every $(t,s,y)\in [0,T]\times \R_+\times \R$
	\begin{enumerate}
		\item \label{ass:verif 1} $\theta[w](t,s,y)\in \K$, recalling the definition in \eqref{eq:opt strategy const lambda}, and
		\item $\mathcal{L}^{\theta[w](t,s,y)}w(t,s,y) = 0$  when $t< T$, and 
		\item \label{eq:verif BC} $w(T, s,y) = H(s,y)$.
	\end{enumerate}
	Suppose further that $w$ is sufficiently regular (see the subsequent remark) so that there exists an admissible strategy $\assetsProcess\in \Gamma$ of the form
	\begin{equation}\label{eq:strategy for smooth sol}
		\begin{aligned}
			\assetsProcess_t &= 1/\lambda \log (\lambda w_s(t, \calS (S_t,\assetsProcess_t), Y_t - \assetsProcess_t)+1) \quad \text{for } t\in [0, T),\\
			\assetsProcess_{T} &= 0,\qquad \hbox{i.e. }\Delta \assetsProcess_T = \assetsProcess_{T-}.
		\end{aligned}
	\end{equation}
	In particular, $\assetsProcess_0 = 1/\lambda \log (\lambda w_s(0, s, y)+1)$ and $\Delta \assetsProcess_T\in\K$. 
	Consider the self-financing portfolio $(\beta, \assetsProcess)$ with $\beta_{0-} = w(0,s,y)$. Then as in \cref{rmk:replicating strategy general f} we get
	\[
	V^{\text{liq}}_T(\assetsProcess) = H(S_T, Y^{\assetsProcess}_T), \qquad \assetsProcess_{T} = 0.
	\]
	By definition of $H$, this shows that $V^{\text{liq}}_T(\assetsProcess)$  at maturity $T$ is enough capital to (super-)\-replicate the European claim with payoff $(g_0, g_1)$ with a possible additional block trade (provided that the infima in the definition of $H$, cf.~\cref{lemma:BC}, are attained). Hence,  $(\beta,\assetsProcess)$ will be a (super-)\-replicating strategy for the European claim $(g_0, g_1)$ with initial capital $w(0,s,y)$, meaning that its price is exactly $w(0,s,y)$.
	
	\begin{remark}[On the form of a replicating strategy]\label{rmk:verification}
		To construct a replicating strategy \eqref{eq:strategy for smooth sol}, suppose moreover that $w\in C^{1,3,1}([0,T]\times \R_+\times \R)$ and apply  It\^{o}'s formula,  similarly as in \cref{rmk:replicating strategy general f}, to get for $t<T$ the equation
		\begin{align*}
			\hspace{-4em}\diff \assetsProcess_t &=\frac{1}{\lambda}\left( \frac{1}{\lambda w_s +1} \diff (\lambda w_s+1) - \frac{1}{2  (\lambda w_s +1)^2}\diff [\lambda w_s +1]_t\right)  \\
			& = a(t, \calS_t, \calY^{\assetsProcess}_t , \assetsProcess_t)\diff t + b(t, \calS_t, \calY^{\assetsProcess}_t) \diff W_t,
		\end{align*}
		where for  $\calS_t := \calS (S_t, \assetsProcess_t)$ and $\calY^{\assetsProcess}_t = Y^{\assetsProcess}_t - \assetsProcess_t$ we set
		\begin{align*}
			a(t, \calS_t, \calY^{\assetsProcess}_t , \assetsProcess_t) &:= \frac{1}{\lambda w_s +1} \Big( w_{ts} + w_{ss}  \calS_t (\mu_t - \lambda h(Y^\assetsProcess_t)) -  w_{sy}h(Y^{\assetsProcess}_t) + \\
			& \qquad \qquad + \frac{1}{2} w_{sss} \sigma^2 \calS_t^2 - \frac{\lambda^2 \sigma^2 \calS_t^2 w_{ss}}{2(\lambda w_{s} + 1)}\Big), \\
			b(t, \calS_t, \calY^{\assetsProcess}_t) &:=  \frac{ \sigma \calS_t w_{ss}}{\lambda w_{s} + 1};
		\end{align*}
		with all the derivatives of $w$ above being evaluated at $(t, \calS(S_t,\assetsProcess_t), Y_t -\assetsProcess_t)$.  
		Thus, a replicating strategy, which is superhedging the payout at a minimal cost, can be constructed as the $(\Theta_t)_{t\in [0,T)}$-part (plus a terminal block trade) from a solution, if it exists, to the SDE system 
		\begin{equation}\label{eq:repl strategy system const lambda}
			\begin{aligned}
				&\diff \calS_t = \calS_t[ (\mu_t - \lambda h(\calY^{\assetsProcess}_t + \assetsProcess_t)) \diff t + \sigma \diff W_t], \\
				&\diff \assetsProcess_t = a(t, \calS_t, \calY^{\assetsProcess}_t,  \assetsProcess_t) \diff t +  b(t, \calS_t, \calY^{\assetsProcess}_t) \diff W_t,\\
				&\diff \calY^{\assetsProcess}_t = -h(\calY^{\assetsProcess}_t + \assetsProcess_t) \diff t
			\end{aligned}
		\end{equation}
		for  $t\in [0,T]$, with initial condition $\calS_0 = s$, $\calY^{\assetsProcess}_0 = y$ and $\assetsProcess_0 = 1/\lambda \log (\lambda w_s(0, s, y)+1)$. 
		
	\end{remark}
	
	\subsection{Viscosity solution property of $w$ for exponential impact function}
	\label{sec:Visc sol property}
	For the result from \cref{subsec:The case of constant lambda}
	we now prove the viscosity property.
	\begin{theorem}\label{thm:supersolution}
		The function $w_*$ from \eqref{eq:def of lower semilimit} is a viscosity supersolution of \eqref{eq:pde const lambda} on $[0,T)\times\RR_+\times\RR$ with the boundary condition \eqref{eq:pde const lambda BC} on $\{T\}\times\RR_+\times\RR$.
	\end{theorem}
	\begin{proof}
		First, let $(t_0, s_0, y_0)\in [0,T)\times\RR_+\times\RR$ and $\varphi\in C^\infty_b([0,T]\times \RR_+\times \RR)$ be a smooth function such that 
		\[
		\text{(strict) } \min_{[0,T]\times \RR_+ \times\RR} (w_* - \varphi) = (w_* - \varphi)(t_0, s_0, y_0) = 0.
		\]

		\textbf{Case 1:}	Suppose that $\H_\K \varphi(t_0, s_0, y_0) < 0$. By continuity of the operator $\H_\K$ there exists an open neighborhood $\mathcal{O}$ of $(t_0, s_0, y_0)$ whose closure is contained in $[0,T)\times \RR_+ \times\RR$, 
		such that $\H_\K \varphi(t, s, y) < -\varepsilon$ in $\mathcal{O}$ for some $\varepsilon > 0$. Therefore, after possibly decreasing the neighbourhood $\mathcal{O}$, there exists a constant $k_{\varepsilon}> 0$ such that
		\begin{equation}\label{eq:case 1 contradiction}
			s|\varphi_S(t,s,y) + 1/\lambda - e^{\lambda \theta}/\lambda| \geq k_\varepsilon \qquad \text{for all } \theta\in \K, (t,s,y)\in \mathcal{O}.
		\end{equation}
		Let $(t_n, s_n, y_n)_n\subset \mathcal{O}$ be a sequence converging to $(t_0, s_0, y_0)$ with $w(t_n, s_n, y_n)\rightarrow  w_*(t_0, s_0, y_0)$  (note that $w_*$ is the lower-semicontinuous envelope of $w$). Set $v_n := w(t_n,s_n,y_n) + 1/n$.
		Since $v_n > w(t_n, s_n, y_n)$, \cref{thm:GDPP} implies the existence of $\theta_n\in \K$ and strategies $\gamma_n\in \Gamma$ such that for stopping times $\tau_n\ge t_n$ (to be suitably  chosen later) we have $\PP$-a.s.
		\begin{equation}\label{eq:case 1 GDPP}
			V^{\text{liq}, t_n, z_n, \gamma_n }_{t\wedge \tau_n} \geq w(\cdot, \calS(S^{t_n,z_n,\gamma_n}, \assetsProcess^{t_n,z_n,\gamma_n}), Y^{t_n,z_n,\gamma_n}-\assetsProcess^{t_n,z_n,\gamma_n})_{t\wedge \tau_n},		
			\quad t\in [t_n,T],
		\end{equation}
		where $z_n = (s_ne^{\lambda \theta_n}, y_n+\theta_n, \theta_n, v_n)$. To abbreviate notation, in the sequel  we  write  $n$ as superscript instead of $(t_n, z_n, \gamma_n)$, with $\calS^n := \calS(S^{t_n,z_n,\gamma_n},\assetsProcess^{t_n,z_n,\gamma_n})$, $\calY^n :=Y^{t_n,z_n,\gamma_n} - \assetsProcess^{t_n,z_n,\gamma_n}$. 
		
		Take $\tau_n = \inf \{t\geq t_n\:\ (t, \calS^n_t, \calY^n_t) \not \in  \mathcal{O}\}$, which is the first entrance time of the  parabolic boundary of the open region $\mathcal{O}$. In particular, $\tau_n < T$. 
		Since $w\geq w_*\geq \varphi$ and $w_* - \varphi$ has a strict local minimum at $(t_0, s_0, y_0)$, there exists $\iota>0$ such that
		\[
		(w- \varphi)(\tau_n, \calS^n_{\tau_n}, \calY^n_{\tau_n}) \geq \iota.
		\]
		Hence, $V^{\text{liq},n}_{\tau_n} - \varphi(\tau_n, \calS^n_{ \tau_n}, \calY^{n}_{ \tau_n})\geq \iota.$ Now, \cref{lemma:dynamics of DPP} together with the fact that $\calS^n_{t_n} = s_n$, $\calY^n_{t_n} = y_n$, gives that $\P$-a.s.
		\begin{align}
			\iota  \leq  v_n  - \varphi(t_n, s_n, y_n)  
			- \int_{t_n}^{\tau_n} \calS^n_u \left(\varphi_S(u, \calS^n_u, \calY^n_u) + 1/\lambda - e^{\lambda \assetsProcess^n_u}/\lambda\right) \left( \sigma \diff W_u + \zeta^n_u\diff u\right) 
			\label{eq:case 1 bound}
		\end{align}
		where
		\[
		\zeta^n_t:= \eta^n_t - \frac{\mathcal{L}^{\assetsProcess^n_t}\varphi}{\calS^n_t (\varphi_S(u, \calS^n_t, \calY^n_t) + 1/\lambda - e^{\lambda \assetsProcess^n_t}/\lambda) } \quad \hbox{for }t\in [t_n, \tau_n]
		\]
		with $\eta^n_t := \mu_t - \lambda h(Y^n_t)$. Note that $\zeta^n_t$ is well-defined on $[t_n, \tau_n]$ and uniformly bounded, noting \eqref{eq:case 1 contradiction} and the fact that $Y^n$ is bounded since $\assetsProcess^n$ is so.
		Hence, by Girsanov's theorem,  there exists a measure $\P^n$ that is equivalent to $\P$ such that 
		\[
		\int_{t_n}^{t\wedge\tau_n} \calS_u (\varphi_S(u, \calS^n_u,\calY^n_u) + 1/\lambda - e^{\lambda \assetsProcess_u}/\lambda) \left( \sigma \diff W_u + \zeta^n_u\diff u\right),
		\quad   t \ge t_n,
		\]
		is a square-integrable martingale under $\PP^n$ 
		as the integrand of the stochastic integral is uniformly bounded, because of the definition of $\tau_n$, the continuity of $\varphi_S$ and the boundedness of the range of $\assetsProcess$, noting $\tau_n\le T$. Taking expectation under $\PP^n$ of the right-hand side of \eqref{eq:case 1 bound} leads to
		\( 
		v_n - \varphi(t_n, s_n, y_n) \geq \iota > 0,
		\) 
		what yields a contradiction as by our choice of $v_n$ and the sequence $(t_n, s_n, y_n)_n$ 
		$$v_n - \varphi(t_n, s_n, y_n) \longrightarrow w_*(t_0, s_0,y_0) - \varphi(t_0, s_0, y_0) = 0.$$
		
		
		\textbf{Case 2:}	From Case 1 we know that $\mathcal{H}_\K \varphi(t_0,s_0,y_0) \geq 0$. Hence
		$$\theta[\varphi](t_0, s_0, y_0) = 1/\lambda \log (\lambda \varphi_S(t_0, s_0, y_0) + 1)$$ is well-defined (also in a neighborhood of $(t_0,s_0,y_0)$). Let us suppose that $\mathcal{L}^{\theta[\varphi]}\varphi(t_0, s_0, y_0) < 0$. By continuity of the operator $\mathcal{L}$, there exits an open neighborhood $\mathcal{O}\subset [0,T]\times \RR_+ \times \RR$ of $(t_0,s_0,y_0)$ and some $r, \varepsilon > 0$ such that
		\begin{align*}
			\mathcal{L}^{\theta}\varphi(t, s, y) < -\varepsilon \quad \text{for all } (t, s, y)\in \mathcal{O},\  \theta \in (\theta[\varphi](t, s, y)-r\ ,\ \theta[\varphi](t, s, y)+r).
		\end{align*}
		In particular, by continuity of the functions involved we have (after possibly decreasing the open set $\mathcal{O}$) that for every $(t,s,y)\in \mathcal{O}$ and for some $r' > 0$ 
		\[
		\mathcal{L}^{\theta}\varphi(t, s, y) < -\varepsilon \quad\text{ whenever }\quad |\varphi_S(t,s,y) + 1/\lambda - e^{\lambda \theta}/\lambda| \leq r'.
		\]
		
		As in Case 1, consider a sequence $(t_n, s_n, y_n)$ in $\mathcal{O}$ which converges to $(t_0, s_0, y_0)$ and such that $w(t_n, s_n, y_n)\rightarrow w_*(t_0, s_0, y_0)$. Set $v_n := w(t_n,s_n,y_n) + 1/n$ and let $\theta_n\in \K$ and strategies $\gamma_n\in \Gamma$ be such that the dynamic programming principle \eqref{eq:case 1 GDPP} holds for the stopping times $\tau_n$ that are the first exit times of $(\cdot, \calS^n, \calY^n)$ from the set $\mathcal{O}$. 
		Now, a contradiction follows similarly as in Case 1 with the following adjustment: We have
		\begin{align*}
			V^{\text{liq},n}_{t\wedge\tau_n}& - \varphi(\cdot, \calS^n, \calY^{n})_{t\wedge \tau_n}  = v_n - \varphi(t_n,s_n,y_n)\\
			& - \int_{t_n}^{t\wedge\tau_n} \calS^n_u (\varphi_S + 1/\lambda - e^{\lambda \assetsProcess^n_u}/\lambda) \left( \sigma \diff W_u + \zeta^n_u\diff u\right) \\
			&+ \int_{t_n}^{t\wedge \tau_n} \mathcal{L}^{\assetsProcess^n_u}\varphi(u, \calS^n_u, \calY^{n}_u) \1_{\{  |\varphi_S + 1/\lambda - e^{\lambda \assetsProcess^n_u}/\lambda| \leq r'\}} \diff u \\
			&\leq  v_n - \varphi(t_n,s_n,y_n) - \int_{t_n}^{t\wedge\tau_n} \calS^n_u (\varphi_S + 1/\lambda - e^{\lambda \assetsProcess^n_u}/\lambda) \left( \sigma \diff W_u + \zeta^n_u\diff u\right),
		\end{align*}
		where we set
		\[
		\zeta^n_t:= \eta^n_t - \frac{\mathcal{L}^{\assetsProcess^n_t}\varphi}{\calS^n_t (\varphi_S + 1/\lambda - e^{\lambda \assetsProcess^n_t}/\lambda) } \1_{\{  |\varphi_S + 1/\lambda - e^{\lambda \assetsProcess^n_t}/\lambda| \geq r' \}}
		\quad \text{for $t\in [t_n, \tau_n]$},
		\]
		with the functions $\varphi$ and $\varphi_S$ in the expressions above being evaluated at $(\cdot, \calS^n_{\cdot}, \calY^n_{\cdot})$. 
		The contradiction now follows by taking expectation under $\P^n\approx \P$  and letting $n\to \infty$.
		
		\textbf{Boundary condition.} Let $(s_0, y_0)\in \RR_+\times\RR$ and $\varphi$ be a smooth function such that 
		\[
		\text{(strict) } \min_{[0,T]\times \RR_+ \times \RR} (w_* - \varphi) = (w_* - \varphi)(T, s_0, y_0) = 0.
		\]
		Suppose that  
		\[
		\min \{w_*(T,s_0,y_0) - H(s_0,y_0), \H_\K\varphi(T,s_0,y_0)\} < 0.
		\]
		The case $\H_\K\varphi(T,s_0,y_0) < 0$ leads to a contradiction by the same arguments as in Case 1 above, using that $\H_\K\varphi < 0$ in a small neighborhood of $(T,s_0,y_0)$. Hence we have $\H_\K\varphi(T,s_0,y_0) \geq 0$.
		
		Now, if $w_*(T,s_0,y_0) < H(s_0,y_0)$ then also $\varphi(T,s_0,y_0) - H(s_0,y_0) < 0$. After possibly modifying the test function $\varphi$ by $(t,s,y)\mapsto \varphi(t,s,y) - \sqrt{T-t}$, we can assume that $\partial_t \varphi(t,s,y)\rightarrow +\infty$ when $t\rightarrow T$, uniformly on compacts. Hence, in an $\varepsilon$-neighborhood $ [T-\varepsilon, T)\times B_\varepsilon(s_0,y_0)$ around $(T, s_0,y_0)$ we have $\mathcal{L}^{\theta[\varphi] }\varphi < 0$. Moreover, after possibly decreasing $\varepsilon$ we have $\varphi(T,\cdot)\leq H(\cdot) - \iota_1$ on $B_\varepsilon(s_0,y_0)$ for some $\iota_1>0$. We can argue as in Cases 1-2 above  (by starting from $(t_n,s_n,y_n)$
		in $[T-\varepsilon,T)\times B_\varepsilon(s_0,y_0)$, with $(t_n,s_n,y_n)\to (T,s_0,y_0)$ and 
		$w(t_n,s_n,y_n)\to w_{*}(T,s_0,y_0)$, stopping at the (parabolic) boundary
		at a time $\tau_n$, and using $w(T, \cdot) = H(\cdot)$)
		to get
		\[
		V^{\text{liq},n}_{\tau_n} - \varphi(\cdot, \calS(S^{n}, \assetsProcess^{n}), Y^{n}-\assetsProcess^{n})_{ \tau_n}\geq \iota_1\wedge \iota_2,
		\]
		where $\iota_2:= \inf_{[T-\varepsilon, T)\times \partial B_\varepsilon(s_0,y_0)} (w_* - \varphi) > 0$.  A contradiction follows as in Case 2 above.

	\end{proof}
	
	Now we prove the subsolution property.
	\begin{theorem}\label{thm:subsolution}
		The function $w^*$ from \eqref{eq:def of upper semilimit} is a viscosity subsolution of \eqref{eq:pde const lambda} on $[0,T)\times\RR_+\times\RR$ with the boundary condition \eqref{eq:pde const lambda BC} on $\{T\}\times\RR_+\times\RR$.
	\end{theorem}
	\begin{proof}
		The proof is similar to and inspired by the one for the subsolution property in \cite[Theorem~3.7]{BLZ16a}. The reason is that in this case, the gradient constraints will ensure that a test function $\varphi$, that would possibly contradict the subsolution property, should satisfy $\H_\K \varphi > 0$ locally and hence would be sufficiently ``nice'' to define (locally) control processes (employing the verification argument in \cref{rmk:verification}) that would lead to a contradiction like in \cite{BLZ16a}. 
		For completeness, we outline differences in the line of proof and sketch the main steps.
		
		Let $\varphi$ be a $C^\infty_b([0,T], \RR_+\times \RR)$ test function such that $(t_0, s_0, y_0)\in [0,T]\times\RR_+\times \RR$ is a strict (local) maximum of $w^* - \varphi$, i.e.
		\[
		\text{(strict) } \max_{[0,T]\times \RR_+ \times \RR} (w^* - \varphi) = (w^* - \varphi)(t_0, s_0, y_0) = 0.
		\]  
		First assume that $t_0 < T$. To ease the notations, we will use the variable $x$ to denote the pair $(s, y)$. Because of the special form of the second part of DPP, cf.~\cref{thm:GDPP} (ii), we need to employ $w_k$ (instead of $w$ as we did in the proof of the supersolution property). By \cite[Lemma~6.1]{Barles2013} we can take a sequence $(k_n, t_n, x_n)_{n\geq 1}$ such that $k_n\rightarrow \infty$, any $(t_n, x_n)$ is a local maximum of $w^*_{k_n} - \varphi$, and $(t_n, x_n, w_{k_n}(t_n, x_n))\rightarrow (t_0, x_0, w^*(t_0, x_0))$.
		
		Assume that $\F_\K[\varphi](t_0, x_0) > 0$ and let $\varphi_n(t,x) = \varphi(t,x) + |t-t_n|^2 + |y-y_n|^2 + |s-s_n|^4$. Then $\F_\K[\varphi_n] > 0$ holds in a neighborhood $B$ of $(t_0, x_0)$ that contains $(t_n, x_n)$, for all $n$ large enough. Since we will be working on the local neighborhood $B$ where also $\mathcal{H}_\K \varphi_n > 0$, we can modify (in a smooth way) the functions $h$ and $\varphi_n$ outside of $B$ to be supported on a slightly bigger compact set where $\mathcal{H}_\K \varphi_n > 0$ holds. Thus, after possibly passing to a suitable subsequence,
		there exists $\gamma_n \in \Gamma_{k_n}$ such that 
		\[
		\assetsProcess^{t_n, z_n, \gamma_n}_t = 1/\lambda \cdot \log \left(\lambda \frac{\partial \varphi_n}{\partial s}(t, \calS^{t_n, z_n, \gamma_n}_t, \calY^{t_n, z_n, \gamma_n}_t) + 1\right),\quad  t\geq t_n,
		\]
		where for $z_n =(s_n, y_n, 0, w_{k_n}(t_n, x_n) - n^{-1})$ we set $\calS^{t_n, z_n, \gamma_n}_t = \calS(S^{t_n, z_n, \gamma_n}_t, \assetsProcess^{t_n, z_n, \gamma_n}_t)$ and $\calY^{t_n, z_n, \gamma_n}_t = (Y-\assetsProcess)^{t_n, z_n, \gamma_n}_t$, see \cref{rmk:verification}.
		Let $\tau_n$ be the first time after $t_n$ at which the process $(\calS^{t_n, z_n, \gamma_n}_t, \calY^{t_n, z_n, \gamma_n}_t)_{t\geq t_n}$ leaves $B$. Like in \cite[proof of Thm.~3.7]{BLZ16a}  we conclude, by applying It\^{o}'s formula, using \cref{lemma:dynamics of DPP} and $F[\varphi_n]> 0$ on $B$, that $\PP$-a.s.
		\begin{align*}
			V^{\text{liq},t_n, z_n, \gamma_n}_{\tau_n}\geq \varphi_n(\tau_n, \calS^{t_n, z_n, \gamma_n}_{\tau_n}, (Y-\assetsProcess)^{t_n, z_n, \gamma_n}_{\tau_n}) + v_n - \varphi_n(t_n, x_n).
		\end{align*}
		Now, a contradiction follows as in \cite[proof of Thm.~3.7, subsolution property, (a)]{BLZ16a}.

		For the boundary condition, i.e.~the case $t_0 = T$, the arguments are exactly the same as in  \cite[proof of Thm.~3.7, subsolution property, (b)]{BLZ16a}.
		
	\end{proof}
	
	\subsection{Comparison results for viscosity solutions}
	\label{sec:comparison}
	
	First we provide a comparison result for the pricing PDE \eqref{eq:pricing pde bounded f}, needed for the proof of \cref{thm:pde for general f}. Note that  \eqref{eq:pricing pde bounded f} has the structure
	\begin{equation}\label{eq:pde structure one asset gen}
		0 = - \varphi_t - \frac{\sigma^2 s^2}{2} \varphi_{ss} - B_1(y, f(y)\varphi_s)\varphi_y -s B_2(y, f(y)\varphi_s)\varphi_s - sB_3(y, f(y) \varphi_s),	
	\end{equation}
	where $B_i:\RR^2 \to \RR$, $i = 1, 2, 3$, are bounded and Lipschitz continuous functions. 
	By a change of coordinates, one can transform the PDE as follows.
	\begin{lemma}\label{lemma:comparison reformulation}
		Let $u$ be viscosity subsolution (resp.~supersolution) of the PDE \eqref{eq:pde structure one asset gen}. Fix $\kappa > 0$. Then the function  $\tilde{u}$, which is defined by
		\[
		\tilde{u}(t,s,y) = e^{\kappa t}u(t,sf(y),y) \quad \text{for all } (t,s,y)\in [0,T]\times \RR_+\times \RR,
		\]
		is a viscosity subsolution (resp.~supersolution)  of the PDE
		\begin{align}
			0 = \kappa \varphi -  \varphi_t -\tfrac{\sigma^2 s^2}{2} \varphi_{ss} -& B_1(y, e^{-\kappa t}\varphi_s)\varphi_y + \lambda(y) B_1(y, e^{-\kappa t} \varphi_s) \varphi_s  \notag\\
			&- s B_2(y,e^{-\kappa t} \varphi_s) \varphi_s - e^{\kappa t} sf(y) B_3(y, e^{-\kappa t}\varphi_s). \label{eq:modified pde comparison}
		\end{align}
	\end{lemma}
	\begin{proof}
		
		To prove super- (resp.~sub-) solution property, let $(t_0, s_0, y_0)\in [0,T)\times\RR_+ \times \RR$ be any point and $\tilde{\varphi}\in C^\infty_b( [0,T]\times \RR_+\times\RR)$ be a test function for $\tilde{u}$ at $(t_0, s_0, y_0)$, i.e.,
		\begin{equation}\label{eq:Lemma A4 test fn 1}
			\min _{[0,T]\times\RR_+\times \RR} (\text{resp.~} \max ) (\tilde u - \tilde{\varphi}) = \tilde u (t_0, s_0, y_0) - \tilde{\varphi} (t_0, s_0, y_0) = 0.
		\end{equation}
		Consider $\varphi (t,s,y) := e^{-\kappa t} \tilde{\varphi}(t, s/f(y), y)$ for all $(t,s,y)\in [0,T]\times \RR_+\times\times\RR$. It holds by definition $e^{\kappa t}\varphi (t,sf(y),y) = \tilde{\varphi}(t,s,y)$ for all $(t,s,y)\in [0,T]\times \RR_+\times\times\RR$. In particular, $\varphi$ is a test function for $u$ at $(t_0, s_0 f(y_0), y_0)$ since by \eqref{eq:Lemma A4 test fn 1} we get
		\begin{equation}\label{eq:Lemma A4 test fn 2}
			\min _{[0,T]\times\RR_+\times \RR} (\text{resp.~} \max ) (u - \varphi) = u (t_0, s_0 f(y_0), y_0) - \tilde{\varphi} (t_0, s_0 f(y_0), y_0) = 0.
		\end{equation}
		We have 
		\begin{align*}
			\tilde{\varphi}_s(t,s,y) &= e^{\kappa t}f(y)\varphi_s(t, sf(y), y) \\
			\tilde{\varphi}_{ss}(t,s,y) &= e^{\kappa t}f^2(y)\varphi_{ss}(t, sf(y), y)\\
			\tilde{\varphi}_y(t,s,y) &= e^{\kappa t}\lambda(y) f(y)\varphi_s(t, sf(y), y) + e^{\kappa t} \varphi_y(t,sf(y), y)\\
			&= \lambda(y) \tilde{\varphi}_s(t,s,y) + e^{\kappa t} \varphi_y(t,sf(y), y) \\
			\tilde{\varphi}_t(t,s,y) &= e^{\kappa t}\varphi_t(t, sf(y), y) + \kappa e^{\kappa t} \varphi(t,sf(y),y).
		\end{align*}
		By direct application of these identities we derive from the right-hand side of \eqref{eq:modified pde comparison} for $\tilde{\varphi}$ evaluated at $(t_0, s_0, y_0)$ exactly the right-hand side of  \eqref{eq:pde structure one asset gen} for $\varphi$ at $(t_0, s_0f(y_0), y_0)$. By the viscosity property of $u$ and \eqref{eq:Lemma A4 test fn 2} we thus conclude that \eqref{eq:modified pde comparison} holds for $\tilde{\varphi}$ at $(t_0, s_0, y_0)$ with greater or equal (resp. less or equal). This proves the claim. 
		
	\end{proof}
	
	By \Cref{lemma:comparison reformulation} it now suffices to prove comparison for equation \eqref{eq:modified pde comparison} since this would imply a comparison result for \eqref{eq:pde structure one asset gen}. This is done in the following result.
	\begin{theorem}\label{thm:comparison multipl general}
		Let  $u$ (respectively~$v$) be a bounded upper-semicontinuous subsolution (resp.~lower-semicontinuous supersolution) on $[0,T)\times \RR_+\times \RR$ of \eqref{eq:modified pde comparison}. Suppose that  $u\le v$ on $\{T\}\times \RR_+\times\RR$. Then $u\le v$ on $[0,T]\times \RR_+\times\RR$.
	\end{theorem}
	\begin{proof}
		To prove the claim by contradiction, let us suppose that
		\[
		\sup_{(t,s,y)\in [0,T]\times \RR_+\times \RR} (u - v)(t,s,y) > 0.
		\] 
		Then we can find $R > 1$ such that 
		$$ \sup_{(t,s,y)\in [0,T]\times \mathcal{O}_R\times [-R, R]} (u - v)(t,s,y) > 0,$$ 
		where $\mathcal{O}_R := (1/R,R)$. In particular, there exists $\delta > 0$ and $(t_0,s_0, y_0)\in \overline{\mathcal{O}}_R\times [-R, R]$ such that $(u - v)(t_0,s_0, y_0) = \delta > 0$.

		Now consider, for $n\in \mathbb{N}$, the bounded upper-semicontinuous function
		\[
		\Phi_n(t,s_1,s_2,y_1, y_2) := u(t,s_1, y_1) - v(t,s_2,y_2) - \frac{n}{2}(s_1-s_2)^2 - \frac{n}{2}(y_1-y_2)^2 .
		\]
		It attains its maximum at some $(t^n, s^n_1, s^n_2, y^n_1, y^n_2)\in [0,T]\times \overline{\mathcal{O}}_R^2\times [-R, R]^2$ by compactness of the set, and we clearly have
		\begin{equation}\label{eq:comp one asset gen contr}
			\Phi_n(t^n,s^n_1,s^n_2,y^n_1, y_2^n) \ge \delta \quad \text{for all } n\in \NN.
		\end{equation}
		By the arguments as in \cite[proof of Lemma~3.11]{BLZ16a} one obtains (after possibly passing to a subsequence)
		that
		\begin{equation}\label{doubling variable single asset gen}
			n(s^n_1-s^n_2)^2 +n (y^n_1-y^n_2)^2\to 0\  \text{ as } n\to \infty.
		\end{equation}
		Note that \eqref{doubling variable single asset gen} also implies $n(s^n_1-s^n_2)(y^n_1-y^n_2) \to 0$ as $n\to \infty$.
		
		Now, by Ishii's lemma as stated in \cite[Theorem~8.3]{viscosity_guide},
		there exist $(b^n, X^n, Y^n)\in \RR\times S_2\times S_2$ such that with $p^n = n(s_1^n - s_2^n)$ and $q^n = n(y_1^n - y_2^n)$ we have
		\begin{align*}
			&(b^n, (p^n, q^n), X^n)\in \bar{\mathcal{P}}_{\O_a}^{2,+} u(t^n, s_1^n, y^n_1),\\
			&(b^n, (p^n, q^n), Y^n)\in \bar{\mathcal{P}}_{\O_a}^{2,-} v(t^n, s_2^n, y^n_2),
		\end{align*}
		where $X^n$ and $Y^n$ satisfy
		\begin{equation}\label{eq:Ishii cond for X and Y}
			\begin{pmatrix}
				X^n & 0\\
				0 	& - Y^n
			\end{pmatrix}
			\leq 
			3n
			\begin{pmatrix}
				I_2 & - I_2\\
				-I_2		&  I_2
			\end{pmatrix},
		\end{equation}
		and where $S_2$ denotes the set of $2\times 2$ symmetric non-negative matrices and  $I_2\in S_2$ is the identity matrix. 
		Using  the viscosity property of $u$ and $v$ at $(t^n, s_1^n, y^n_1)$ and $(t^n, s_2^n, y^n_2)$ respectively, we have
		\begin{align*}
			&\kappa u(t^n, s_1^n, y^n_1)  - b_n - \tfrac{1}{2}\sigma^2  (s_1^n)^2 X^n_{11} + L(s_1^n, y_1^n, p^n, q^n) \le 0 \\
			&\kappa v(t^n, s_2^n, y^n_2)  - b_n - \tfrac{1}{2}\sigma^2  (s_2^n)^2 Y^n_{11} + L(s_2^n, y_2^n, p^n, q^n)\ge 0,
		\end{align*}
		where 
		\[
		L(t,s,y,p,q) := - B_1(y, e^{-\kappa t}p) q + \lambda(y) B_1(y, e^{-\kappa t}p) p - s B_2(y,e^{-\kappa t}p) p - e^{\kappa t} sf(y) B_3(y, e^{-\kappa t}p). 
		\] 
		As a consequence,
		\begin{align}
			0<\kappa \delta &< \kappa(u(t^n, s_1^n, y^n_1) - v(t^n, s_2^n, y^n_2))  \leq \notag\\
			&\le  - \tfrac{1}{2}\sigma^2  (s_2^n)^2 Y^n_{11} +  \tfrac{1}{2}\sigma^2  (s_1^n)^2 X^n_{11} + \notag \\
			&\qquad \qquad+ L(t^n, s_2^n, y_2^n, p^n, q^n) - L(t^n, s_1^n, y_1^n, p^n, q^n).
			\label{eq:comparison singl ass gen contr}
		\end{align}
		On the other hand, by \eqref{eq:Ishii cond for X and Y} we get that  
		$$\tfrac{1}{2}\sigma^2  (s_1^n)^2 X^n_{11} - \tfrac{1}{2}\sigma^2  (s_2^n)^2 Y^n_{11} \le \tfrac{3}{2}\sigma^2 n(s_1^n - s_2^n)^2,$$
		what  converges to $0$ for $n\to \infty$ due to \eqref{doubling variable single asset gen}.
		Let us now analyze the difference $ L(t^n, s_2^n, y_2^n, p^n, q^n)$ $- L(t^n, s_1^n, y_1^n, p^n, q^n) $.
		With  $C$ (resp.~$C_R$) denoting a Lipschitz constant (depending on $R$), that may change from line to line,
		we get estimates for the corresponding terms as follows:
		\[
		|B_1(y^n_1, e^{-\kappa t^n}p^n) q^n -  B_1(y^n_2, e^{-\kappa t}p^n) q^n|  \le C |y^n_1 - y^n_2||q^n|,
		\]
		\[
		| \lambda(y_1^n) B_1(y_1^n, e^{-\kappa t^n}p^n) p^n - \lambda(y_2^n) B_1(y_2^n, e^{-\kappa t^n}p^n) p^n|  \le C |y_1^n - y_2^n||p^n|,
		\]
		\[
		|s^n_1 B_2(y^n_1,e^{-\kappa t^n}p^n) p^n - s^n_2 B_2(y^n_2,e^{-\kappa t^n}p^n) p^n|  \leq  C|(s^n_1 - s^n_2) p^n| + C_{R}|(y^n_1 - y^n_2)p^n|,
		\]
		\[
		|e^{\kappa t^n} s_1^nf(y_1^n) B_3(y_1^n, e^{-\kappa t^n}p^n) - e^{\kappa t^n} s_2^n f(y_2^n) B_3(y_2^n, e^{-\kappa t^n}p^n)| \leq C_{R}( |s_1^n - s_2^n| + |y_1^n - y_2^n|).
		\]
		As all estimates from above 
		vanish for $n\to \infty$, the right-hand side in \eqref{eq:comparison singl ass gen contr} is bounded by something that converges to $0$ as $n\to \infty$. But this yields a contradiction for large $n$.
	\end{proof}

	Because of lack of a precise reference, we provide a comparison result also in the case of delta constraints leading to the variational inequality \eqref{eq:pde const lambda}.
	
	\begin{theorem}\label{thm:comparison}
		Suppose that the resilience function $h$ is Lipschitz continuous and \cref{as:semicontinuous limits are finite} holds.	Let $u$ (resp.~$v$) be bounded upper- (resp.~lower-) semicontinuous viscosity subsolution (resp.~supersolution) of the variational inequality \eqref{eq:pde const lambda} with the terminal condition \eqref{eq:pde const lambda BC}. Then $u\le v$ on $[0,T]\times \RR_+ \times \RR$.
	\end{theorem}
	\begin{proof}
		
		We argue by contradiction. For any $a > 0$, set  $\O_a := [a,\infty) \times [-1/a, 1/a]$. Suppose that 
		\(\sup_{(t,s,y)\in [0,T]\times \RR_+\times \RR} (u - v) > 0.\)
		Then there exists some $a > 0$ such that $\sup_{(t,s,y)\in [0,T]\times \O_a} (u - v) > 0$. For $\kappa > 0$, consider $\tilde{u}:= e^{\kappa t}u$ and $\tilde{v}:= e^{\kappa t}v$. Then $\tilde{u}$ (resp.~$\tilde{v}$) is a viscosity sub- (resp.~super-) solution of 
		\begin{align*}
			\min \{\kappa \varphi + \tilde{\mathcal{L}} [\varphi], \mathcal{H}_{\K,t}\varphi  \} = 0
		\end{align*}	 
		with the boundary condition $\min\{\varphi(T, \cdot) - H(\cdot), \mathcal{H}_{\K,T} \varphi\} = 0$, where 
		\[
		\tilde{\mathcal{L}} [\varphi](t, s, y) = -\partial_t \varphi + h(y + 1/\lambda\log(\lambda e^{-t\kappa}\partial_s\varphi + 1)) \partial_y\varphi - 1/2 \sigma^2 s^2 \partial_{ss} \varphi
		\]
		and $\mathcal{H}_{\K,t}\varphi = \lambda e^{-\kappa t}\partial_s\varphi + 1 - e^{-\lambda K}$ for $t\in [0,T]$.
		
		Consider 
		\[
		\Theta_n := \sup_{(t,x_1,x_2)\in [0,T]\times \O_a^2} \tilde{u}(t,x_1) - \tilde{v}(t,x_2) - \frac{n}{2}|x_1-x_2|^2. 
		\]	
		We have $\Theta_n > \iota$ for some $\iota>0$. Since $\tilde{u} - \tilde{v}$ is upper-semicontinuous,
		for any $n$ the supremum is attained    as a  maximum   at some $(t_n, x^n_1, x^n_2)$ 
		in the compact set $[0,T]\times \O_a^2$. By  arguments as in \cite[proof of Lemma~3.11]{BLZ16a}, after possibly passing to a subsequence, we obtain
		\begin{align}
			\label{eq*conclusion 1.}
			&  \text{$\lim_{n\rightarrow\infty}\Theta_n = \sup_{(t,s,y)\in [0,T]\times \O_a} (\tilde{v} - \tilde{u}) \geq \iota >  0$, and}
			\\
			\label{eq*conclusion 2.}
			&  \text{ $n|x_1^n- x_2^n|^2 \rightarrow 0$ as $n\rightarrow \infty$.}
		\end{align}
		Note also  that 
		\begin{equation}\label{eq:comp contrad 1}
			\lim_{n\rightarrow \infty}\tilde{u}(t_n, x_1^n) - \tilde{v}(t_n, x_2^n) \geq \iota.
		\end{equation}
		
		\textbf{Case 1: } Suppose, after passing to a subsequence, that $t_n = T$ for all $n$. Then Ishii's lemma together with the viscosity property of $\tilde{u}$ and $\tilde{v}$ give
		\begin{align*}
			&\min\left\{ \tilde{u}(T, x_1^n) - H(x_1^n), \ \lambda e^{-\kappa T} p_n + 1 - e^{-\lambda K}\right\} \leq 0,\\
			&\min\left\{ \tilde{v}(T, x_2^n) - H(x_2^n), \ \lambda e^{-\kappa T} p_n + 1 - e^{-\lambda K}\right\} \geq 0,
		\end{align*}
		where $p_n = n(s_1^n - s_2^n).$
		Hence we conclude that $ \tilde{u}(T, x_1^n) \leq H(x_1^n)$ for all $n$. However, in this case since $ \tilde{v}(T, x_2^n) \geq  H(x_2^n)$ for all $n$ we have 
		\[
		\tilde{v}(T, x_2^n) \geq  H(x_2^n) \geq H(x_2^n) - H(x_1^n) + \tilde{u}(T, x_1^n),
		\]
		which contradicts \eqref{eq:comp contrad 1} for large $n$ by continuity of $H$.
		
		\textbf{Case 2: } We can now assume (after passing to a subsequence) that $t_n < T$ for all $n$. Set
		\[
		p_n := n(s^n_1 - s^n_2),\quad  q_n := n(y^n_1 - y^n_2).
		\]
		By Ishii's lemma, see \cite[Theorem~8.3]{viscosity_guide}, using the viscosity property of $\tilde{u}$ and $\tilde{v}$, there exist $a_n\in \RR$ and symmetric $2\times 2$ matrices $A_n, B_n$ (that satisfy a bound like in \eqref{eq:Ishii cond for X and Y})  with
		\[
		(a_n, (p_n, q_n), A_n)\in \bar{\mathcal{P}}_{\O_a}^{2,+} \bar{u}(t_n, x_1^n), \ \ (a_n, (p_n, q_n), B_n)\in \bar{\mathcal{P}}_{\O_a}^{2,-} \bar{v}(t_n, x_2^n),
		\] 
		such that
		\begin{align*}
			&\min\left\{ -a_n +  L(t_n, x_1^n,\tilde{u}(t_n, x_1^n), p_n,q_n, A_n), \ \lambda e^{-\kappa t_n} p_n + 1 - e^{-\lambda K}\right\} \leq 0,\\
			&\min\left\{ -a_n + L(t_n, x_2^n,\tilde{v}(t_n, x_2^n), p_n,q_n, B_n), \ \lambda e^{-\kappa t_n} p_n + 1 - e^{-\lambda K}\right\} \geq 0,
		\end{align*}
		where for $t\in [0,T]$, $x = (x_1, y_1)\in \RR^2$, $\ell, p,q\in \RR$ and a $2\times 2$ matrix $A$, we define
		\[
		L(t, x = (x_1, y_1), \ell, p,q, A) :=  \kappa \ell  + h(y_1 + 1/\lambda \log(\lambda e^{-\kappa t}p + 1))q - 1/2 \sigma^2 x_1^2 A_{11}.
		\]
		Therefore, we have
		$-a_n +  L(t_n, x_1^n,\tilde{u}(t_n, x_1^n), p_n, q_n, A_n) \leq 0.$
		
		On the set $\{(t, y, p)\in [0,T]\times \RR\times \RR\mid \lambda e^{-\kappa t} p + 1 - e^{-\lambda K} \geq 0\}$, the function
		\[
		(t, y, p)\mapsto h(y + 1/\lambda \log(\lambda e^{-\kappa t}p + 1))
		\]
		is Lipschitz continuous. Thus, we are exactly in the setup the proof of \cref{thm:comparison multipl general} and a contradiction argument yields the claim like in the proof there:  One gets the estimate
		\[
		\kappa(\tilde{u}(t_n, x_1^n) - \tilde{v}(t_n, x_2^n)) \leq C \left( n|x_1^n - x_2^n|^2   + 1/n\right)
		\]
		for a constant $C> 0$ not depending on $n$, what contradicts 
		\eqref{eq*conclusion 1.} for $n$ large.
	\end{proof}
	
	\begin{remark}\label{rmk:proof of uniqueness}
		By Theorems \ref{thm:supersolution} and \ref{thm:subsolution},
		$w_*$ (resp.~$w^*$) is a supersolution (subsolution) of \eqref{eq:pde const lambda} with boundary condition  \eqref{eq:pde const lambda BC}. By \cref{thm:comparison}, we have $w_*\geq w^*$ on $[0,T]\times \RR_+\times \RR$. It is clear by definition that $w_*\leq w^*$, hence $w_* = w^*$ on $[0,T]\times \RR_+\times \RR$. On the other hand, $w_*\leq w \leq w^*$ on 
		$[0,T)\times \RR_+\times \RR$. To show equality also at $t = T$, note that the super-/sub-solution property of $w_{*}$/$w^{*}$, respectively, 
		implies $w_*(T,\cdot)\ge H(\cdot)$ and $w_*(T,\cdot)\le H(\cdot)$. So $w_*$ equals $H$ at $T$. Since 
		also $H(\cdot) = w(T,\cdot)$, the equality $w_* = w^* = w$ holds on $\{T\}\times \RR_+\times \RR$. Hence $w_* = w^* = w$ on  $[0,T]\times \RR_+\times \RR$, what implies continuity.
		
		The same conclusion holds for  \eqref{eq:pricing pde bounded f}  with the boundary condition \eqref{eq:def of H}.
	\end{remark}

\end{appendix}

{\footnotesize 

\bibliographystyle{spmpsci}  		

\bibliography{paper}

\begin{thebibliography}{10}
\providecommand{\url}[1]{{#1}}
\providecommand{\urlprefix}{URL }
\expandafter\ifx\csname urlstyle\endcsname\relax
  \providecommand{\doi}[1]{DOI~\discretionary{}{}{}#1}\else
  \providecommand{\doi}{DOI~\discretionary{}{}{}\begingroup
  \urlstyle{rm}\Url}\fi

\bibitem{AckermannEtal22OW}
Ackermann, J., Kruse, T., Urusov, M.: Reducing {Obizhaeva-Wang} type trade
  execution problems to {LQ} stochastic control problems.
\newblock arXiv preprint arXiv:2206.03772  (2022)

\bibitem{AlfonsiSchiedSlynko12}
Alfonsi, A., Schied, A., Slynko, A.: Order book resilience, price manipulation,
  and the positive portfolio problem.
\newblock SIAM J. Financial Math. \textbf{3}(1), 511--533 (2012)

\bibitem{BankBaum04}
Bank, P., Baum, D.: Hedging and portfolio optimization in financial markets
  with a large trader.
\newblock Math. Finance \textbf{14}(1), 1--18 (2004)

\bibitem{BankSonerVoss16}
Bank, P., Soner, H.M., Vo{\ss}, M.: Hedging with temporary price impact.
\newblock Math. Financ. Econ. pp. 1--25 (2016)

\bibitem{Barles2013}
Barles, G.: An Introduction to the Theory of Viscosity Solutions for
  First-Order Hamilton–Jacobi Equations and Applications, p. 49–109.
\newblock Springer Berlin Heidelberg, Berlin, Heidelberg (2013).
\newblock \doi{10.1007/978-3-642-36433-4_2}

\bibitem{BilarevBecherer18}
Becherer, D., Bilarev, T.: Hedging with transient price impact for non-covered
  and covered options.
\newblock arXiv preprint arXiv:1807.05917  (2018)

\bibitem{BechererBilarevFrentrup2016-deterministic-liquidation}
Becherer, D., Bilarev, T., Frentrup, P.: Optimal asset liquidation with
  multiplicative transient price impact.
\newblock Appl. Math. Optim. \textbf{78}(3), 643--676 (2018)

\bibitem{BechererBilarevFrentrup2016-stochastic-resilience}
Becherer, D., Bilarev, T., Frentrup, P.: Optimal liquidation under stochastic
  liquidity.
\newblock Finance Stoch. \textbf{22}(1), 39--68 (2018)

\bibitem{BechererBilarevFrentrup-model-properties}
Becherer, D., Bilarev, T., Frentrup, P.: Stability for gains from large
  investors' strategies in {M1}/{J1} topologies.
\newblock Bernoulli \textbf{25}(2), 1105--1140 (2019)

\bibitem{BertsimasLo98}
Bertsimas, D., Lo, A.W.: Optimal control of execution costs.
\newblock J. Financial Markets \textbf{1}(1), 1--50 (1998)

\bibitem{Bilarev18}
Bilarev, T.: Feedback effects in stochastic control problems with liquidity
  frictions.
\newblock Ph.D. thesis, Humboldt-Universität zu Berlin (2018).
\newblock \doi{http://dx.doi.org/10.18452/19592}

\bibitem{BLZ16a}
Bouchard, B., Loeper, G., Zou, Y.: Almost-sure hedging with permanent price
  impact.
\newblock Finance Stoch. \textbf{20}(3), 741--771 (2016)

\bibitem{BLZ16b}
Bouchard, B., Loeper, G., Zou, Y.: Hedging of covered options with linear
  market impact and gamma constraint.
\newblock SIAM J. Control Optim. \textbf{55}(5), 3319--3348 (2017)

\bibitem{BouchardTan22}
Bouchard, B., Tan, X.: Understanding the dual formulation for the hedging of
  path-dependent options with price impact.
\newblock Ann. Appl. Prob. \textbf{32}(3), 1705--1733 (2022)

\bibitem{BussetiLillo12}
Busseti, E., Lillo, F.: Calibration of optimal execution of financial
  transactions in the presence of transient market impact.
\newblock J. Stat. Mech. Theory Exp. \textbf{2012}(09), P09010 (2012)

\bibitem{CetinSonerTouzi10}
{\c{C}}etin, U., Soner, H.M., Touzi, N.: Option hedging for small investors
  under liquidity costs.
\newblock Finance Stoch. \textbf{14}(3), 317--341 (2010)

\bibitem{CEK15_delta}
Chassagneux, J.F., Elie, R., Kharroubi, I.: When terminal facelift enforces
  delta constraints.
\newblock Finance Stoch. \textbf{19}(2), 329--362 (2015)

\bibitem{viscosity_guide}
Crandall, M.G., Ishii, H., Lions, P.L.: User's guide to viscosity solutions of
  second order partial differential equations.
\newblock Bull. Amer. Math. Soc. (N.S.) \textbf{27}(1), 1--67 (1992)

\bibitem{Frey1998}
Frey, R.: Perfect option hedging for a large trader.
\newblock Finance Stoch. \textbf{2}(2), 115--141 (1998)

\bibitem{FreyPolte11}
Frey, R., Polte, U.: Nonlinear {B}lack--{S}choles equations in finance:
  Associated control problems and properties of solutions.
\newblock SIAM J. Control Optim. \textbf{49}(1), 185--204 (2011)

\bibitem{GuoZervos13}
Guo, X., Zervos, M.: Optimal execution with multiplicative price impact.
\newblock SIAM J. Financial Math. \textbf{6}(1), 281--306 (2015)

\bibitem{HorstKivman21}
Horst, U., Kivman, E.: Optimal trade execution under small market impact and
  portfolio liquidation with semimartingale strategies.
\newblock arXiv preprint arXiv:2103.05957  (2021)

\bibitem{HubermanStanzl04}
Huberman, G., Stanzl, W.: Price manipulation and quasi-arbitrage.
\newblock Econometrica \textbf{72}(4), 1247--1275 (2004)

\bibitem{Jarrow92}
Jarrow, R.A.: Market manipulation, bubbles, corners, and short squeezes.
\newblock J. Financial Quant. Anal. \textbf{27}(3), 311–336 (1992)

\bibitem{KolmWebster23}
Kolm, P.N., Webster, K.: Do you really know your {P\&L}? {The} importance of
  impact-adjusting the {P\&L}.
\newblock SSRN preprint 4331027  (January 19, 2023).
\newblock \doi{10.2139/ssrn.4331027}

\bibitem{ObizhaevaWang13}
Obizhaeva, A., Wang, J.: Optimal trading strategy and supply/demand dynamics.
\newblock J. Financial Markets \textbf{16}, 1--32 (2013)

\bibitem{PredoiuShaikhetShreve11}
Predoiu, S., Shaikhet, G., Shreve, S.: Optimal execution in a general one-sided
  limit-order book.
\newblock SIAM J. Financial Math. \textbf{2}(1), 183--212 (2011)

\bibitem{SchoenbucherWilmott2000}
Sch{\"o}nbucher, P.J., Wilmott, P.: The feedback effect of hedging in illiquid
  markets.
\newblock SIAM J. Appl. Math. \textbf{61}(1), 232--272 (electronic) (2000)

\bibitem{SonerTouzi02}
Soner, H.M., Touzi, N.: Dynamic programming for stochastic target problems and
  geometric flows.
\newblock J. Eur. Math. Soc. \textbf{4}(3), 201--236 (2002)

\end{thebibliography}
}
\end{document}